\documentclass[12pt]{article}
\usepackage{graphicx}
\usepackage{enumitem} 
\usepackage{amsthm} 
\usepackage{float} 
\usepackage{array}
\usepackage{makecell} 
\usepackage{pbox} 
\usepackage{subcaption}
\usepackage{xcolor}
\usepackage{soul} 
\newcommand{\abs}[1]{\left\vert #1 \right\vert}
\newcommand{\pa}[1]{\left( #1 \right)}
\newcommand{\be}{\begin{equation}}
\newcommand{\ee}{\end{equation}}

\newtheorem{thm}{Theorem}[section]
\newtheorem{lem}[thm]{Lemma}
\newtheorem{rem}[thm]{Remark}
\usepackage{indentfirst} 
\usepackage{amsmath}
\usepackage{amssymb}

\usepackage{geometry}
\usepackage[utf8]{inputenc}
\usepackage[section] {placeins} 
\usepackage{hyperref}
\usepackage{tabularx,ragged2e,booktabs}
\usepackage{authblk}
\usepackage{tcolorbox} 
\usepackage{algorithm} 
\usepackage[noend]{algpseudocode}

\begin{document}

\title{Evaluating interventions for \textit{Plasmodium vivax} forest malaria using a three-scale mathematical model}

\author[1,2]{Shoshana Elgart}
\author[3]{Mark B. Flegg}
\author[4]{Jennifer A. Flegg}
\affil[1]{Laurel Springs School, Ojai, California, United States}
\affil[2]{Stanford University, Stanford, California, United States}
\affil[3]{School of Mathematics, Monash University, Melbourne Australia}
\affil[4]{School of Mathematics and Statistics, The University of Melbourne, Parkville, Australia}
\date{}                      
\setcounter{Maxaffil}{0}
\renewcommand\Affilfont{\itshape\small}
\maketitle
\begin{abstract}
The rising proportion of \textit{Plasmodium vivax} cases concentrated in forest-fringe areas across the Greater Mekong Subregion highlights the importance of pharmaceutical and mosquito control techniques specifically targeted towards forest-going populations. To mathematically assess best-possible antimalarial interventions in the context of hypnozoite reactivation and seasonal forest migration, we extend a previously developed three-scale integro-differential equations model of \textit{P. vivax} transmission. In particular, we fit the model to data gathered over a  four-year period in Vietnam to gain insight into local \textit{P. vivax} dynamics and validate the model's ability to capture epidemiological trends. The calibrated model is then used to generate optimal schedules for mass-drug administration (MDA) in forest-goers and gauge the efficacy of vector control techniques (such as long-lasting insecticide nets and indoor residual spraying) in forest-adjacent areas. Our results highlight the dependence of optimal MDA timing on the demographics of the human population, the importance of interventions targeting the mosquito bite rate, and the need for efficacy in hypnozoite-targeting antimalarial drugs. 
\end{abstract}

\section{Introduction}\label{Intro}
In the past several decades, increasing urbanization and a renewed focus on antimalarial interventions (e.g., mass-drug administration, mobile healthcare stations, and vector control measures such as source reduction) have considerably diminished the malaria caseload in Southeast Asia. In Cambodia, for instance, the per-capita incidence of malaria in 2019 was over 3.5 times less than that recorded in 2006; and in Vietnam, incidence decreased by over 95\% between 1991 and 2014 \cite{chhim2021malaria, wangdi2018analysis}. As of the most recent WHO Malaria Report, no country in the Greater Mekong Subregion (GMS), however, has achieved elimination, with the exception of China in 2021. Cambodia, for instance, numbered approximately eighteen thousand malaria cases in 2022, a sizable annual per-capita incidence of over 0.1\% \cite{world2023world}. 

An increasing proportion of malaria caseload may be classified as \textit{Plasmodium vivax} in origin: unlike the parasites associated with \textit{Plasmodium falciparum}, the second predominant cause of malaria in Southeast Asia, \textit{P. vivax} parasites are distinguished by a dormant stage, in which parasites can live as hypnozoites in liver cells. These hypnozoites can activate sporadically, entering the erythrocytic stage and inducing symptomatic malaria and infectiousness. 

Critical to the sustained transmission of \textit{P. vivax} in the GMS is cyclic migration. \textit{P. vivax} remains endemic in forest-fringe areas, including Oddar Meanchey and Stung Treng Province in Cambodia, Champasak Province in Laos, and Binh Phuc Province in Vietnam, chiefly due to the predominance of the logging, hunting, and tracking industries there, which require individuals, typically termed ``forest-goers", to work overnight in high-risk forested areas  \cite{jongdeepaisal2021acceptability, kounnavong2017malaria, masunaga2021search}. In Cambodia, for instance, forest-going groups face an annual per-capita malaria incidence of over 37.7\%, stressing the continued brunt of malaria for communities involved in forest-based work \cite{iv2024intermittent}. Epidemiological malaria studies focusing on forest-going populations have largely been empirical, either qualitative and interview-based or quantitative and risk-assessing. These studies have stressed the intensified importance of antimalarials and personal protective measures in forest-adjacent communities, describing the need to improve  the accessibility and efficacy of mass-drug administration (MDA) for forest-going people in the GMS \cite{bannister2019forest, jongdeepaisal2021acceptability, nofal2019can, phok2022behavioural, sandfort2020forest}.

Despite the literature identifying forest-going as critical to the spread of malaria, antimalarial interventions for forest-goers in the GMS have been unevenly applied, made complicated both by environmental conditions in the forest and by the epidemiological complexity of \textit{P. vivax}. Vector control techniques such as indoor residual spraying, topical repellents, and insecticide-treated nets were found to be ineffective for individuals moving across large open spaces, while long-sleeved protective clothing was sometimes found impractical, and thus less adhered to, in the  heat and humidity of the forest \cite{nofal2019can,wilson2014topical}.
Fully treating \textit{P. vivax} infections requires eliminating within-liver hypnozoites, but the key anti-hypnozoital drugs primaquine and tafenoquine (also known as radical cure) should not be administered to individuals with G6PD deficiency, so as to prevent oxidant haemolysis \cite{world2016testing}. Furthermore, it is unclear how to seasonally time the use of MDA, in which entire communities---including both susceptible and infectious individuals---receive doses of antimalarial drugs, to forest travel. 

A notable body of mathematical and epidemiological work has studied interventions for \textit{P. vivax}, though few studies have specifically considered \textit{P. vivax} risk in cyclically migrating individuals. In particular, White \textit{et al.} \cite{white2018mathematical} forecasted reductions in \textit{P. vivax} caseload in the Solomon Islands and Papua New Guinea in response to upticks in insecticide-treated net use, while Nekkab \textit{et al.} \cite{nekkab2021estimated} considered the effect of pharmaceutical interventions with tafenoquine radical cure in Brazil. In 2023, Anwar \textit{et al.} \cite{anwar2023optimal} used a combined within-host and population-level model of \textit{P. vivax} (based on the work of Mehra \textit{et al.} \cite{Mehra} and Anwar \textit{et al.} \cite{Anwar}) to determine the best-possible timing of $N$ rounds of MDA in a given population. 
Spatial models of forest-fringe villages include Gerardin \textit{et al.} \cite{gerardin2018impact}, which used an agent-based simulation to optimize seasonal MDA timing for \textit{P. falciparum} (as opposed to \textit{P. vivax}) and Li \textit{et al.} \cite{li2023understanding}, another agent-based model illustrating the relationship between \textit{P. vivax} caseload and forest travel in Myanmar. Building on the multiscale model of \textit{P. vivax} developed in \cite{Mehra} and \cite{Anwar} (which incorporated both the within-host components of hypnozoite accrual and activation, as well as the population-level dynamics of transmission between humans and mosquitoes), Elgart \textit{\textit{et al.}} \cite{pastpaper} introduced a third, spatial scale, modelling \textit{P. vivax} transmission in a village with a significant forest-going population. 

In this paper, we use the framework of the three-scale metapopulation model developed in Elgart \textit{et al.} \cite{pastpaper} to theoretically characterize best-possible \textit{P. vivax} interventions in mobile populations. To our knowledge, we  derive the first analytical approach for optimizing the timing of MDA in forest-going groups. Simultaneously, we use sensitivity analyses and model fitting of the integro-differential equations model in \cite{pastpaper} (to data gathered in Vietnam by Wangdi \textit{et al.} \cite{wangdi2018analysis}) to garner  insight into the epidemiological dynamics of \textit{P. vivax} transmission in forest-adjacent areas.

 \section{Model development} \label{model}

In this section, we briefly re-develop the metapopulation \textit{P. vivax} model first derived in Elgart \textit{et al.} \cite{pastpaper}. We begin by describing the within-host scale embedded into the model, originally described in the work of Mehra \textit{et al.} \cite{Mehra}. 

\subsection{Within-host model}\label{withinhost}
We model the dynamics of a single hepatic hypnozoite, inoculated at time $t = 0$ (Figure \ref{fig1}). Mehra \textit{et al.} \cite{Mehra} defined the hypnozoite to inhabit one of four potential states: the original \textit{establishment} state $H$, the \textit{activated} state $A$, the \textit{cleared} state $C$, and the \textit{hypnozoital death} state $D$. Quiescent hypnozoites in $H$ may either activate or die prior to activation: The transition from $H$ to $A$, completed at rate $\alpha$, represents the process of activation, in which the hypnozoite becomes a schizont and triggers a secondary infection in the host; this process is assumed to be rapid in tropical environments, meaning that no intermediate compartments between establishment and activation are required. The transition from $H$ to $D$, which occurs at rate $\mu$, represents parasite death. Finally, the transition from $A$ to $C$ models recovery in the host (i.e., the clearance of the activated hypnozoite), and occurs at rate $\gamma$, where all rates are defined in days$^{-1}$. Mehra \textit{et al.} \cite{Mehra} established the following equations for $p_H, p_A, p_C, p_D$, the probabilities that the hypnozoite is found in states $H, A, C$, or $D$, respectively, at time $t$:

\begin{figure}
     \centering
\includegraphics[width=0.8\textwidth]{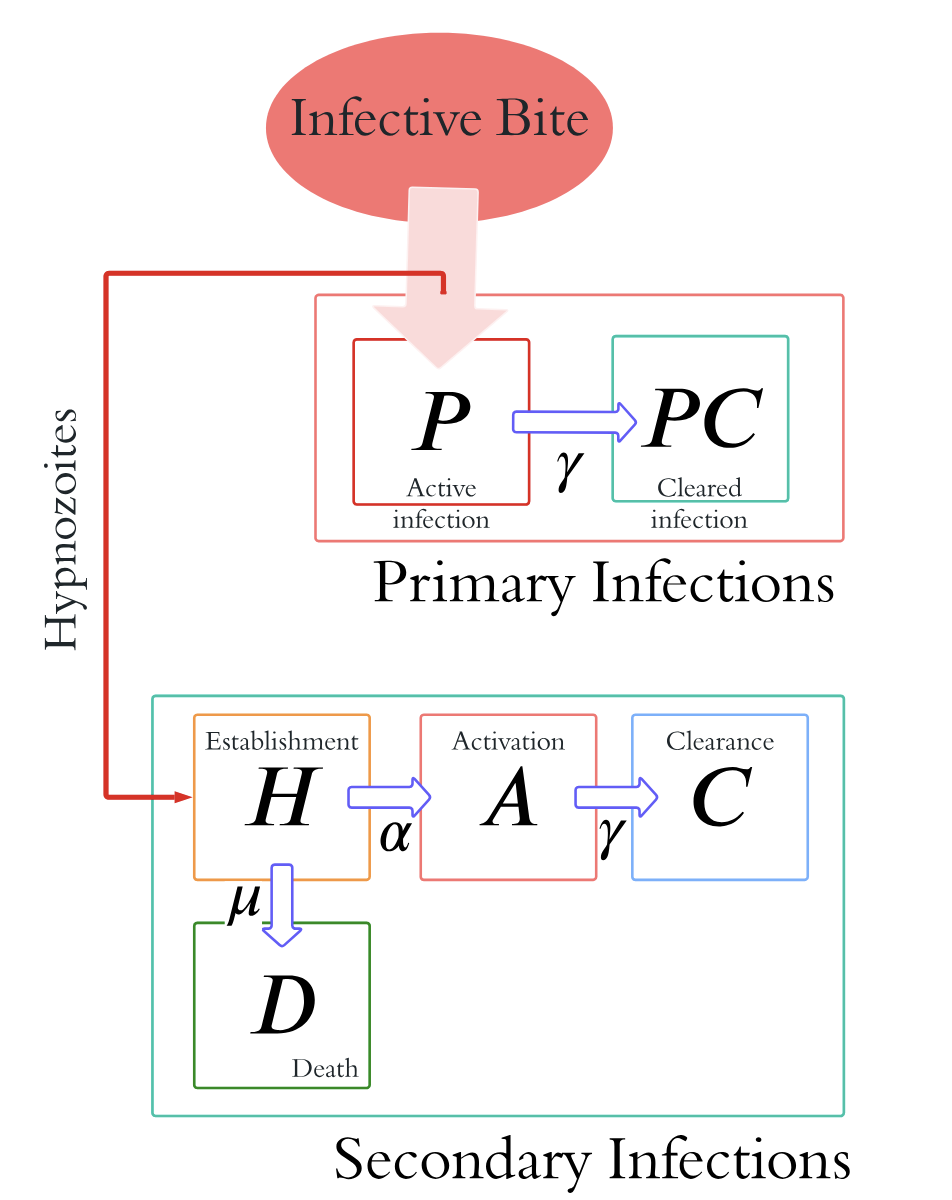}
     \caption{Schematic illustrating the components associated with the within-host model for a single infectious bite, which both generates an active primary infection and inoculates a particular quantity of hypnozoites (with mean $\nu$). Adapted from Figure 2 in Elgart \textit{et al.} \cite{pastpaper}.}   \label{fig1}
\end{figure} 

\be\label{hypnozoiteEqs} \begin{aligned} & p_H(t) = e^{-\pa{\alpha + \mu}t}; \\
& p_{A}(t) = \frac{\alpha}{\alpha + \mu - \gamma}\pa{e^{-\gamma t} - e^{-(\alpha + \mu) t}}; \\ 
& p_C(t) = \frac{\alpha}{\alpha + \mu}\pa{1 - e^{-(\alpha + \mu)t}} - \frac{\alpha}{\alpha + \mu - \gamma} \pa{e^{-\gamma t} - e^{-(\alpha + \mu) t}}; \\
&  p_D(t) = \frac{\mu}{\alpha + \mu} \pa{1 - e^{-(\alpha + \mu)t}}.  \end{aligned} \ee

The dynamics of each hypnozoite are assumed to be independent of the others, such that one can derive a probability generating function (PGF) for the distribution of hypnozoites across compartments at time $t$, conditional on an infectious bite at time $\tau$. In particular, from Equation (49) in Mehra \textit{et al.} \cite{Mehra}, we obtain that   
\[\mathbb{E}\Bigr[\prod_{j \in J_h} z_j^{\mathcal{N}_j(t)} \mid A_\tau \Bigr] = \frac{1}{1 + \nu \pa{1 - \sum_{j \in J_h} z_j \cdot\, p_j(t - \tau)}},\] where $\nu$ is the mean number of hypnozoites established in the liver per infectious bite, $J_h$ is the set of compartments $\{H, A, C, D\}$, $\mathcal{N}_j(t)$ denotes the cardinality of compartment $j$ for $j \in J_h$, and $A_\tau$ is the event of the (only) infectious bite occurring at  time $\tau$.

Simultaneously, under the assumption that each \textit{P. vivax}-carrying mosquito bite instantaneously results in a primary infection, two additional compartments, denoted $P$ and $PC$, are introduced to represent the states associated with non-hypnozoite-induced hematocyte infections. Here, $P$ contains active infections, while $PC$ contains cleared ones. The transition from an active primary infection to a cleared one is assumed to occur at the same rate as the elimination of a secondary infection, namely $\gamma$ (see Figure \ref{fig1}). 

We now let $J = J_h \cup \{P, PC\}$, retaining the same definition for $\mathcal{N}_j(t)$, as above. Equation (30) in \cite{Mehra} defines a PGF for the distribution of hypnozoites and primary infections across states within a single individual at time $t$, with infectious bites modelled as a Poisson process over the interval $(0, t]$:

\begin{equation}\label{eq:pgforiginal}
\begin{aligned} 
& G(z_H, z_A, z_C, z_D, z_P, z_{PC}) :=   \\ & \mathbb{E} [\prod_{j \in J} z_{j}^{\mathcal{N}_j(t)}] = \exp \int_0^t  \lambda(\tau)\pa{\pa{\frac{\pa{z_{P}e^{-\gamma(t - \tau)} + (1 - e^{-\gamma(t - \tau)})z_{PC}}}{1 + \nu\pa{1 - \sum_{j \in J_h} z_{j} \cdot p_{j}(t - \tau)}} \hspace{3pt}} - 1} d\tau.
\end{aligned}
\end{equation}
Here, $\lambda$ is the force of (re-)infection (FORI) to which the individual is exposed, or the rate at which a previously susceptible individual becomes infectious upon exposure to mosquito bites. We express the FORI in terms of  the product of the average \textit{Anopheles} bite rate, the proportion of infectious mosquitoes, and the probability of transmission given a bite from an infected mosquito. In particular, the definition of $\lambda$ is made precise in Section \ref{deriv:population} below.

Equation \eqref{eq:pgforiginal} may be used to determine the probability $p_I(t)$ that an individual in an environment with FORI $\lambda(t)$ is infectious at any given time. In particular, infectiousness  at time $t$ corresponds to the case when $\mathcal{N}_A(t) + \mathcal{N}_P(t) > 0$; using this condition with Equation \eqref{eq:pgforiginal} yields that \be \label{inf} p_I(t) = 1 - P\pa{\mathcal{N}_A(t) + \mathcal{N}_P(t) = 0} = 1 - G\pa{1, 0, 1, 1, 0, 1} = 1- \exp \pa{\int_0^t \lambda(\tau) f(t - \tau)\, d \tau}\ee where we define $f(t)$ by \be\label{f_Eq} f(t) := \frac{1 - e^{- \gamma t}}{1 + \nu p_{A}(t)} - 1,
 \ee 
 noting that $-1 \le f(t) < 0$, $f$ is monotonically increasing, and $\lim_{t \to \infty} f(t) = 0$. A slightly modified version of this computation is found in Equation (15) in Anwar \textit{et al.} \cite{Anwar}.
 
 We notice that, as in Anwar \textit{et al.} \cite{anwar2023optimal}, this derivation for $p_I(t)$ leaves open the possibility of \textit{superinfection} in the population, or the presence of multiple sets of activated hypnozoites and erythrocytic parasites cohabiting in an infectious human, since we do not place restrictions on the total cardinality $\mathcal{N}_A(t) + \mathcal{N}_P(t)$.

\subsection{Population-level model}\label{population-level}

In this section, we present an extension of the within-host model above to the population and metapopulation-level scales, summarizing the model development sections in Anwar \textit{et al.} \cite{Anwar} and Elgart \textit{et al.} \cite{pastpaper}. In particular, we let $\mathcal{H}$ denote the set of humans, and $\mathcal{Q}$ denote the set of mosquitoes. To represent forest travel, we define the \textit{P. vivax} transmission and infection dynamics along two patches (Figure \ref{Two-patch}). Patch 1, referred to as the ``village" patch, is the patch of permanent residence for all individuals in $\mathcal{H}$, associated with a localized and non-migrating population of mosquitoes $\mathcal{Q}_1 \subset \mathcal{Q}$. Patch 2, referred to as the ``forest" patch, is associated with a (similarly localized and non-migrating) population of mosquitoes $\mathcal{Q}_2 := \mathcal{Q} \setminus \mathcal{Q}_1$. Here, the populations are such that $\abs{Q_2} >> \abs{Q_1}$, since forest \textit{Anopheles} populations, less encumbered by LLINs and other insecticide interventions, typically far exceed village populations in number \cite{overgaard2003effect, white2014modelling}. As such, letting $\lambda_i(t)$ denote the FORI that individuals on Patch $i$ experience at time $t$, we assume that $\lambda_2(t) >> \lambda_1(t)$.

To generalize the cyclic partial population migration seen in forest travel, we identify individuals in $\mathcal{H}$ as belonging to either a ``moving" or a ``stationary" (non-forest going) group. In particular, at regular intervals of time specified below, Patch 2 is frequented by the moving subset of the human population, denoted $\mathcal{M} \in \mathcal{H}$, while the remaining humans (in the stationary group $\mathcal{N} := \mathcal{H} \setminus \mathcal{M}$) remain on Patch 1. We suppose that both $|\mathcal{M}|$ and $|\mathcal{N}|$ are arbitrarily large. Simultaneously, we assume that the model represents a short time-scale (e.g., a year or less), meaning that immigration may be disregarded, and demographic processes assumed negligible. 

To consider the number of moving and stationary individuals as bounded proportions of the total human population, we introduce the notation \[M:= \frac{|\mathcal{M}|}{|\mathcal{H}|}; \hspace{10pt} N:= \frac{|\mathcal{N}|}{|\mathcal{H}|}.\] In particular, $N + M = 1,$ and both $N$ and $M$ are constant. 

For the mosquito populations, we introduce the ratios \[Q_1 := \frac{|\mathcal{Q}_1|}{\abs{\mathcal{H}}}; \hspace{10pt} Q_2 := \frac{|\mathcal{Q}_2|}{\abs{\mathcal{H}}},\] 
where we presume that $Q_1 << 1$.

We now consider more precisely the movement processes by which individuals migrate between patches (Figure 2). We first define the \textit{movement functions} $A_1(t), A_2(t):= \mathbb{R}_+ \to [0, 1]$, which give the fraction of moving individuals on Patches 1 and 2, respectively, at time $t$. More specifically, letting the quantity of moving individuals in $\mathcal{M}$ on Patch $i$ be denoted by the function $\mathcal{M}_i(t)$, we have that \[A_i(t) = \frac{\abs{\mathcal{M}_i(t)}}{\abs{\mathcal{M}}}.\] As such, $A_1(t) + A_2(t) = 1$ and the proportion of humans on Patch 1 at time $t$ becomes $N + A_1(t)M$, while the proportion of humans on Patch 2 must be $A_2(t)M$. Since forest-going in the GMS follows seasonal trends (as suggested, for instance, by Bannister-Tyrrell \textit{\textit{et al.}} \cite{bannister2019forest}), we choose the form
\be\label{A1A2} \begin{aligned} A_1(t) = u + v \cos\pa{\omega t}, \\ A_2(t) = 1 - u - v \cos\pa{\omega t},\end{aligned} \ee
for the movement functions $A_1, A_2$, defining a small ``amplitude" parameter $v << 1$, a larger parameter $u$ satisfying $0.5 \leq u < 1 - v$ (which represents the mean proportion of moving individuals on Patch 1 over time), and the parameter $\omega \le 2\pi/365$. In particular, $\omega$ models the average frequency of forest travel, which we allow to be arbitrary in our mathematical results (but, as discussed in our numerical results, is typically taken to be such that the period of the population-wide travel cycle encompasses a single year).  

Below, in Section \ref{eq:deriv}, we will use the equations presented here (alongside the result summarized in Equation (\ref{inf})) to determine the equation for the compartment of infectious moving individuals over time, denoted $I_M(t)$. Similarly, we use the within-host model outlined in Section \ref{withinhost} to derive the equation for the compartment of infectious stationary individuals, denoted $I_N(t)$. As in Elgart \textit{\textit{et al.}} \cite{pastpaper}, these equations do not depend on susceptible or liver-only infected compartments; only $I_M(t)$ and $I_N(t)$. Instead, $I_M(t)$ and $I_N(t)$ can be expressed as Volterra convolution equations, determined exclusively by the history of the FORI on one or both patches and the parameters governing within-host and population dynamics.

\begin{figure}[h!]
     \centering
         \includegraphics[width=\textwidth]{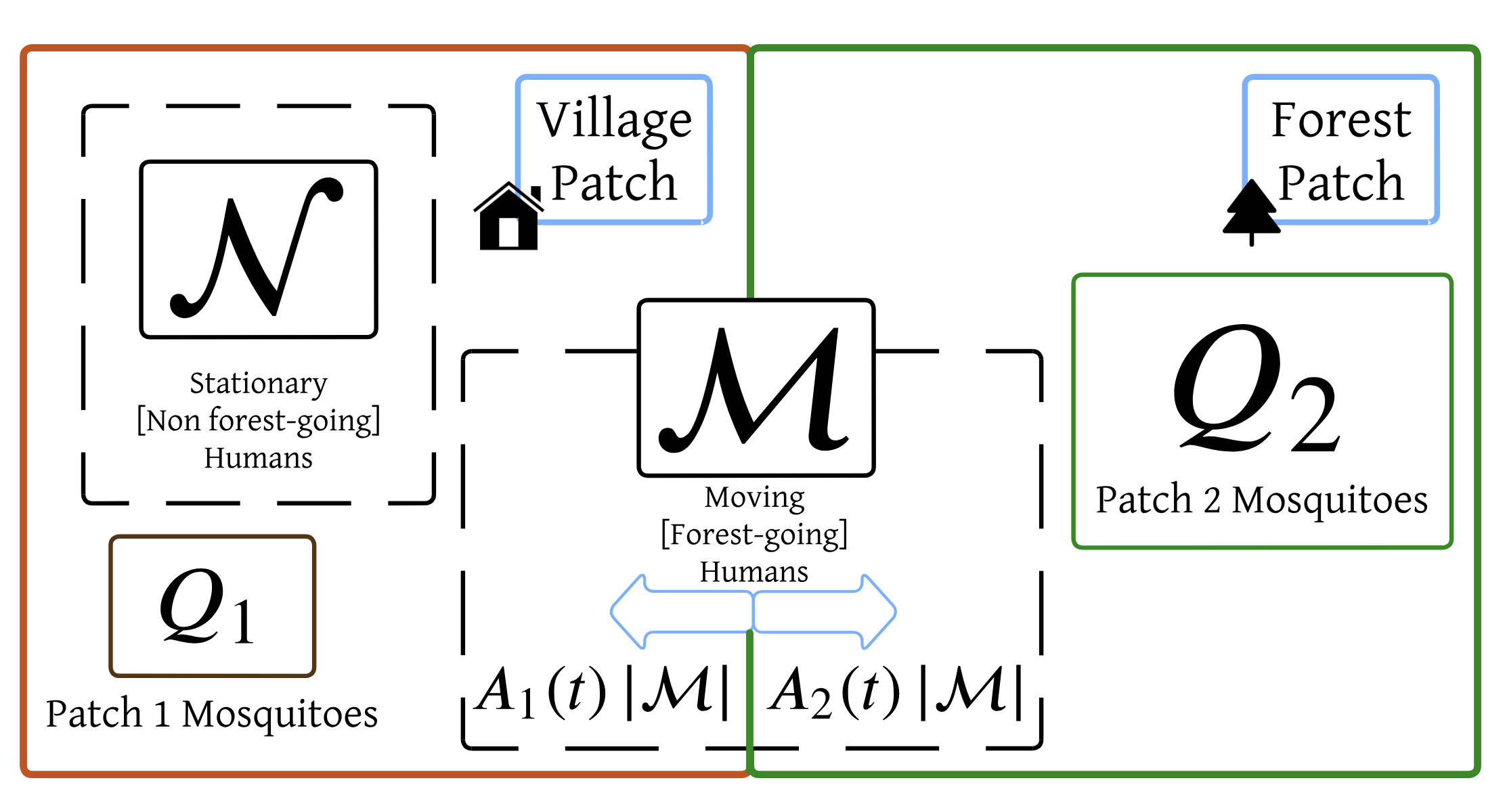}
 \caption{Schematic of the two-patch metapopulation model, with Patch 1 (the village patch) on the left and Patch 2 (the forest patch) on the right. Each patch is associated with a distinct and non-moving population of mosquitoes, denoted $\mathcal{Q}_1$ and $\mathcal{Q}_2$, respectively; non-forest going humans (in $\mathcal{N}$) remain stationary on the village patch, while forest-going humans (in $\mathcal{M}$) traverse the patches, with $A_i(t)\abs{\mathcal{M}}$ forest-going humans located on Patch $i$, where $i \in \{1, 2\}$, at any fixed time $t$ (such that the proportion of the human population found non-permanently on Patch $i$ is $A_i(t)M$). Adapted from Figure 1 in Elgart \textit{et al.} \cite{pastpaper}.}         \label{Two-patch}

\end{figure}

 \subsubsection{Volterra integral equations for the infectious human compartments}
 \label{eq:deriv}

 Since the stationary population is homogeneous with respect to hypnozoite and infection dynamics and the number of stationary humans is assumed to be arbitrarily large, by the Law of Large Numbers, the density of $I_N(t)$ approaches the probability that a single stationary individual is infectious. This probability is given for a non-spatial model in Equation (\eqref{inf}) and hence we have
 \be \label{ineq} I_N(t) = N - N \exp\pa{\int_{0}^{t}\lambda_1(\tau)\, f(t - \tau) \, d \tau}.\ee 

 Similarly, we can closely approximate the density of infectious moving individuals in the population by the expected size of $I_M(t)$.  
 To calculate $\mathbb{E}[I_M(t)]$, we let there be a bijection mapping each 
 moving person $j \in \mathcal{M}$ to a periodic, deterministic \textit{movement function} $B_j(t)$, which functions as an indicator variable, equal to $1$ while the corresponding individual remains on Patch $1$, and to $0$ when the individual is found on Patch $2$.  Fixing an individual with movement function $B_j(t)$, where $j \in \mathcal{M}$, we notice that the PGF capturing within-host dynamics for this individual takes the form 
 \begin{equation}\label{eq:pgfmoving}
\begin{aligned} 
& G(\tilde{z}_H, \tilde{z}_A, \tilde{z}_C, \tilde{z}_D, \tilde{z}_P, \tilde{z}_{PC}) :=   \\ & \mathbb{E} [\prod_{j \in J} \tilde{z}_{j}^{\mathcal{N}_j(t)}] = \exp \int_0^t  \pa{B_j(\tau) \lambda_1(\tau) + (1 - B_j(\tau)) \lambda_2(\tau)} \times \\ & \pa{\frac{\tilde{z}_{P}e^{-\gamma(t - \tau)} + (1 - e^{-\gamma(t - \tau)})\tilde{z}_{PC}}{1 + \nu\pa{1 - \sum_{j \in J_h} \tilde{z}_{j} \cdot p_{j}(t - \tau)} \hspace{3pt}} - 1} d\tau
\end{aligned}
\end{equation} (as presented in Equation (10) of \cite{pastpaper}).  
 By extending Equation (\ref{inf}) in Section \ref{withinhost}, we find that the probability of an individual with this movement function to be in $I_M$ at time $t$ must therefore be
 \[p_j(t) = 1 - \exp\pa{\int_{0}^{t}  \pa{B_j(\tau) \lambda_1(\tau) + (1 - B_j(\tau)) \lambda_2(\tau)}f(t - \tau)\ d\tau}.\]

We assume that the movement functions associated with each forest-going individual are shifted versions of one another (where the extent of the shift may vary with time). In particular, we suppose that \[B_j(t) = B_i(t + w_{ij}(t))\] for all  $i, j \in \mathcal{M}$, where each function $w_{ij}(t)$ satisfies $\abs{w_{ij}} << 1$ (i.e., there is little to no variation in individual movement). 

In Appendix II, we prove the following theorem:

\begin{thm}\label{thmEq}
For arbitrarily small values of $Q_1, |v|,$ and $\textnormal{sup}_{i,j\in\mathcal{M}} \abs{w_{i,j}}$, arbitrarily large values of $\omega$, and $\abs{f(t)}<1$, we have that \be I_M(t) = M - M \pa{\frac{1}{1 + \epsilon}} \exp\pa{\int_{0}^{t} \Bigr[A_1(\tau) \lambda_1(\tau) + A_2(\tau)  \lambda_2(\tau)\Bigr]f(t - \tau)\ d\tau},\ee
where $\abs{\epsilon} << 1$. As such, given these assumptions, $I_M(t)$ can closely be approximated by \begin{equation}\label{densityofmovinginf} \begin{aligned} I_M(t) \approx M - M \exp\pa{\int_{0}^{t} \pa{A_1(\tau) \lambda_1(\tau) + A_2(\tau) \lambda_2(\tau)}f(t - \tau)\ d\tau}.\end{aligned} \end{equation}
\end{thm}

\begin{rem}
    The only feature of $f(t)$ needed to prove Theorem \ref{thmEq} is the fact that $|f(t)|$ is bounded everywhere by $1$. (In particular, the theorem would hold if $f(t)$ was replaced by any function $|g(t)| \le |f(t)|$.) We use this observation in Section \ref{mda} to optimize the population-level (rather than individual) timing for a single round of MDA. 
\end{rem}

\subsubsection{A precise definition for the force of reinfection (FORI)}\label{deriv:population}
We define $\lambda_i(t)$, where $i \in \{1, 2\}$, as the number of infecting bites per capita, per day, at time $t$, on Patch $i$. This in turn is given by the product of the ratio of infectious mosquitoes to humans on Patch $i$ and the number of \textit{P. vivax}-transmitting bites per infectious \textit{Anopheles} individual per day. 

We denote the proportion of infectious mosquitoes relative to the number of humans on Patch $i$ by $I_{mi}(t)$, $i \in \{1, 2\}$. In particular (analogously to our definitions of $M, N, Q_1$, and $Q_2$), if $\mathcal{I}_{mi}(t)$ denotes the set of malaria-carrying mosquitoes localized to Patch $i$ at time $t$, where $i \in \{1, 2\}$, then \[I_{mi}(t) = \frac{\abs{\mathcal{I}_{mi}(t)}}{\abs{\mathcal{H}}}.\] As such, $I_{mi}(t)/Q_i$ gives the fraction of the \textit{Anopheles} population on Patch $i$ that belongs to an infectious state at time $t$. 

As in the case of the notation for the density of infectious humans, we refrain from distinguishing between $\mathcal{I}_{mi}(t)$ and $I_{mi}(t)$ in the below, instead following the convention of using $I_{m1}(t)$ and $I_{m2}(t)$ to additionally represent the \textit{compartments} of infectious mosquitoes. Furthermore, we let $a$ denote the expected value of the bite rate per mosquito, and let $b$ denote the probability of transmission from mosquitoes to humans given an infective bite.    

As such, following Equation (11) in \cite{pastpaper}, we obtain that 
\begin{equation}\label{lambdaeq} \begin{aligned} \lambda_1(t) = \underbrace{ab}_{\substack{ \text{Mean number} \\ \text{of infecting} \\ \text{bites per} \\ \text{mosquito} \\ \text{a day}}} \cdot \ \underbrace{\frac{I_{m1}(t)}{Q_1}  \cdot  \frac{Q_1}{\pa{N + A_1(t) \, M}}}_{\substack{\text{Infectious mosquito} \\ \text{to human ratio}}} = \frac{ab\,I_{m1}(t)}{\pa{N + (u + v \cos \omega t) \, M}}; \\ \lambda_2(t) = ab\ \cdot \ \frac{I_{m2}(t)}{Q_2} \ \cdot  \ \frac{Q_2}{A_2(t)M} = \frac{ab\,I_{m2}(t)}{M \pa{(1-u) - v \cos \omega t}}. \end{aligned} \end{equation}

\subsubsection{A complete system of equations for the metapopulation model}

In this section, we complete the re-derivation of the model in Elgart \textit{\textit{et al.}} \cite{pastpaper} by deriving the equations for the mosquito compartments, which are coupled to Equations (\ref{ineq}) and (\ref{densityofmovinginf}). In particular, following the structure of the model in Anwar \textit{\textit{et al.}} \cite{Anwar}, we assume that mosquitoes on any given patch are either susceptible, exposed (during which state the sporogonic cycle takes place within the mosquito host), or infectious and liable to transmit the disease (Figure \ref{Schematic_3}). 

Analogously to the use of $I_{m1},I_{m2}$ to represent the proportion of infectious mosquitoes relative to the total number $\abs{\mathcal{H}}$ of humans, we let $S_{m1}$ and $S_{m2}$ represent the density of susceptible mosquitoes on the forest and village patches, respectively, as a fraction of $\abs{\mathcal{H}}$. Similarly, we let $E_{m1}$ and $E_{m2}$ represent the proportions of mosquitoes in a latent state on Patches $1$ and $2$, respectively, relative to the total quantity $\abs{\mathcal{H}}$. In particular, we note that, for all $t$,
\be\begin{aligned}
    & S_{m1}(t) + E_{m1}(t) + I_{m1}(t) = Q_{1}; \\
    & S_{m2}(t) + E_{m2}(t) + I_{m2}(t) = Q_{2}. 
\end{aligned} \ee

As in Anwar \textit{et al.} \cite{Anwar}, mosquitoes transition from the susceptible state to the infectious state at a rate given by the mosquito force of infection (mFOI) associated with each patch. The mFOI are extensions of the human FORI, differing only in that the density of infectious humans and probability of human-to-mosquito transmission is considered, rather than the density of infectious mosquitoes and the probability of mosquito-to-human transmission.   In Equation (21) of Elgart \textit{\textit{et al.}} \cite{pastpaper}, the mFOI are given by \begin{equation}\begin{aligned}& \Phi_1(t) = \frac{ac \, \pa{I_{N} + (u - v \cos \omega t) \, I_{M}}}{N + (u + v \cos \omega t) \, M}; \\ & \Phi_2(t) = \frac{ ac\,I_M}{M},
\end{aligned}\end{equation} where $\Phi_i(t), i \in \{1, 2\}$ denotes the mFOI on Patch $i$ at time $t$ and $c$ is the probability of transmission given a susceptible mosquito biting an infectious human. 
 The transition from the exposed state to the infectious one occurs at rate $n$, while mosquito demography (i.e., both birth and death) occurs at the rate $g_i$ on Patch $i$. 
 
 In the original derivation of the population-level and three-scale models (\cite{Anwar}, \cite{pastpaper}), the parameters $g_i$ were assumed time-invariant. (In the case of the three-scale model, this ensured that periodic human migration remained the only possible cause of observed cyclicity in malaria prevalence.) Here, however, as in Anwar \textit{et al.} \cite{anwar2023optimal}, we relax the assumption that $g_1$ and $g_2$ are constant, so as to better fit the model to real-life data from Vietnam, where seasonal effects on mosquito egg-laying and larval development rates likely account for a significant proportion of oscillatory behavior in annual prevalence and incidence \cite{wangdi2018analysis}.
 
Following Anwar \textit{et al.} \cite{anwar2023optimal}, we introduce the following sinusoidal temporal dependence into the expressions for $g_1$ and $g_2$:  

\be\label{eq:humsimplify:seasonality} g_1(t) := \tilde{g}_1 + \tilde{g}_1\eta\,\cos\pa{2\pi t/365}; \hspace{10pt} g_2(t) := \tilde{g}_2 + \tilde{g}_2 \eta \,\cos \pa{2 \pi t/ 365}.\ee

Here, $\eta$ is the \textit{seasonality amplitude} parameter, specifying the extent of seasonal fluctuations in mosquito demography, while $\tilde{g}_1, \tilde{g}_2$ are the constant \textit{baseline} demography parameters for mosquitoes on Patches 1 and 2, respectively. We fit all three of these parameters to existing incidence data below. 

The model equations are therefore given by:  
\begin{figure}[h!]
     \centering
         \includegraphics[width=\textwidth]{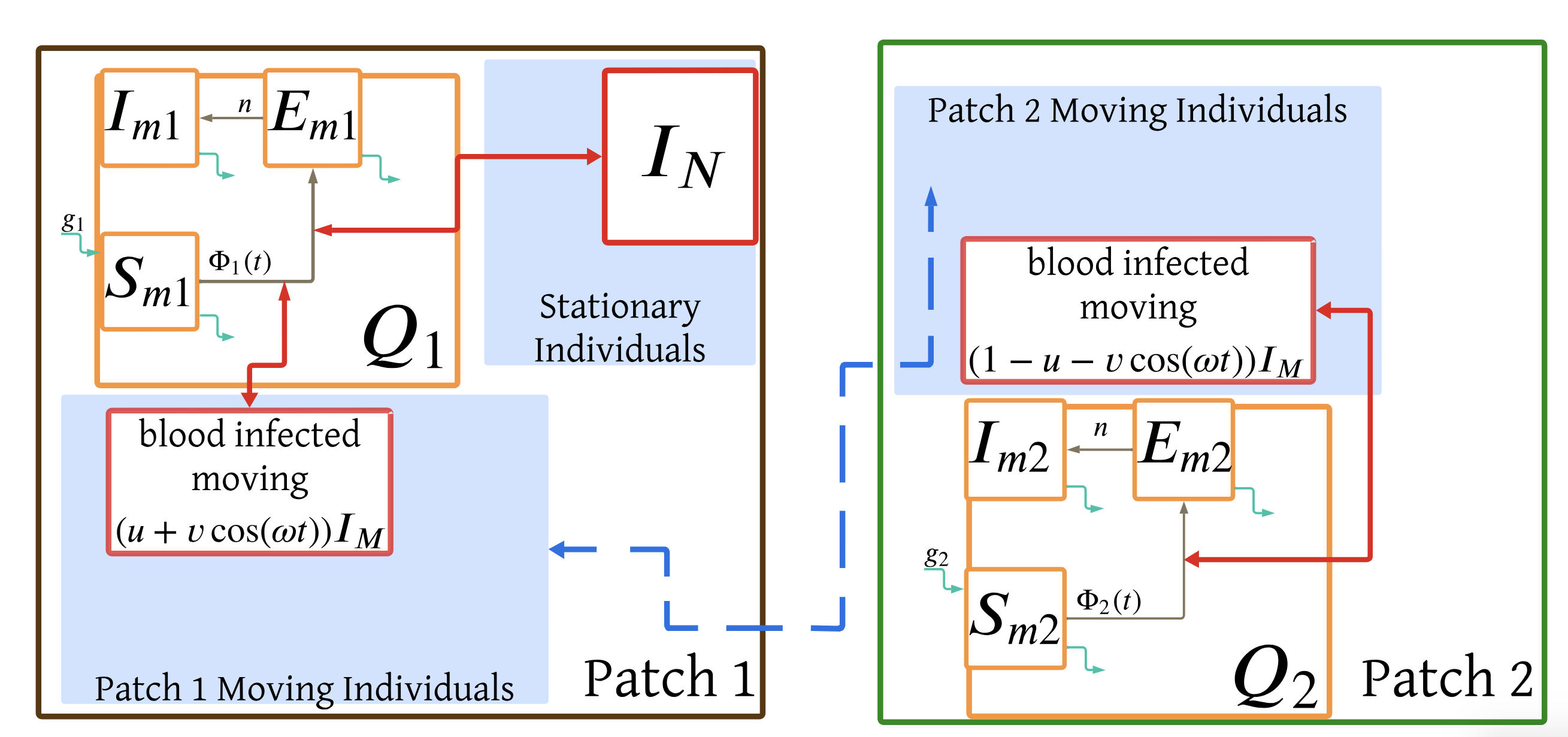}
     \caption{Schematic showing the transitions between mosquito and human compartments across Patch 1 and Patch 2. The blue dashed arrow between the patches represents forest travel while red arrows represent infection mechanisms. Adapted from Figure 4 in Elgart \textit{et al.} \cite{pastpaper}.}\label{Schematic_3}
\end{figure} 

\begin{equation}\label{eq:humsimplify}
\begin{aligned}
    &\frac{d S_{m1}}{dt} = g_{1}(t)Q_1 - \frac{ac \, \pa{I_{N} + (u + v \cos \omega t) \, I_{M}}}{N + (u + v \cos \omega t) \, M}\, S_{m1} - g_{1}(t) S_{m1}; \\
    & \frac{d E_{m1}}{dt} = \frac{ac \, \pa{I_{N} + (u +  v \cos \omega t) \, I_{M}}}{N + (u + v \cos \omega t) \, M}\, S_{m1} - (g_{1}(t) + n)E_{m1}; \\
    & \frac{d I_{m1}}{dt} = n E_{m1} - g_{1}(t) I_{m1}; \\
    &\frac{d S_{m2}}{dt} = g_{2}(t)Q_2 - \frac{ ac\,I_M}{M} S_{m2} - g_{2}(t) S_{m2}; \\
    & \frac{d E_{m2}}{dt} = \frac{ ac\,I_M}{M} S_{m2} - (g_{2}(t) + n)E_{m2}; \\
    & \frac{d I_{m2}}{dt} = n E_{m2} - g_{2}(t) I_{m2};\\
    & I_{M} = M  - M \exp\pa{ab\,\int_{0}^{t}  \pa{\frac{I_{m1}(\tau)(u + v \cos \omega \tau)}{N + (u + v \cos \omega \tau)M}  } f(t - \tau)\, d\tau} \times \\ &  \hspace{100pt} \exp\pa{ab\, \int_0^t \frac{I_{m2}(\tau)((1 - u) - v\cos \omega \tau)}{((1 - u) - v\cos \omega \tau)M} \, f(t - \tau)\, d\tau}; \\
    & I_{N} =  N - N \exp\pa{ab\,\int_{0}^{t}\frac{I_{m1} (\tau)}{N +  (u + v\cos \omega \tau) \, M}\, f(t - \tau) \, d \tau}. \\
  \end{aligned}
\end{equation}
 As in Equation (23) of \cite{pastpaper}, the initial conditions are: 
\be \begin{aligned} & S_{m1}(0) = S^{(0)}_{m1}; \hspace{3pt} E_{m1}(0) = E^{(0)}_{m1}; \hspace{3pt}I_{m1}(0) = I^{(0)}_{m1}; \\ &   \hspace{3pt}S_{m2}(0) = S^{(0)}_{m2}; \hspace{3pt}E_{m2}(0) = E^{(0)}_{m2}; 
\hspace{3pt}I_{m2}(0) = I^{(0)}_{m2}; \hspace{3pt} I_M(0) = 0; \hspace{3pt} I_N(0) = 0. \end{aligned} \ee
\subsubsection{Relevant analytical results for the model in Equation \eqref{eq:humsimplify}}\label{approx}

In \cite{pastpaper}, we determine numerically that the endemic solutions of the model in Equation (\ref{eq:humsimplify}) approach periodic orbits for reasonable regions of parameter space (where $\eta$ has been set to zero, such that the only periodic forcing in the system occurs due to movement from patch to patch). As such, the limiting solutions \\ $S^{*}_{m1}, S^{*}_{m2}, E^{*}_{m1}, I^{*}_{m1}, S^{*}_{m2}, E^{*}_{m2}, I^{*}_{m2}, I^{*}_M, I^{*}_N$ are sought under the ansatz  

\begin{equation}\label{coeff}
\begin{aligned}
& S^{*}_{mi}(t) \approx u^{(i)}_{-1} \exp \pa{-\omega \textbf{i} t} + u^{(i)}_0 + u^{(i)}_{1} \exp \pa{\omega \textbf{i} t} + u^{(i)}_\epsilon, \hspace{5pt} i \in \{1, 2\}; \\
& E^{*}_{mi}(t)  \approx v^{(i)}_{-1} \exp \pa{-\omega \textbf{i} t} + v^{(i)}_0 + v^{(i)}_{1} \exp \pa{\omega \textbf{i} t} + v^{(i)}_\epsilon, \hspace{5pt} i \in \{1, 2\};  \\
& I^{*}_{mi}(t) \approx w^{(i)}_{-1} \exp \pa{-\omega \textbf{i} t} + w^{(i)}_0 + w^{(i)}_{1} \exp \pa{\omega \textbf{i} t} + w^{(i)}_\epsilon, \hspace{5pt} i \in \{1, 2\};  \\
& I^{*}_{M}(t)  \approx x_{-1} \exp \pa{-\omega \textbf{i} t} + x_0 + x_{1} \exp \pa{\omega \textbf{i} t} + x_\epsilon;  \\
& I^{*}_{N}(t)  \approx y_{-1} \exp \pa{-\omega \textbf{i} t} + y_0 + y_{1} \exp \pa{\omega \textbf{i} t} + y_\epsilon, \\
\end{aligned}
\end{equation}
where the $\epsilon$-order terms are assumed to be arbitrarily small. 
Under this approximation, the risk of infection for an unspecified moving or stationary individual varies sinusoidally over each ``movement period" $[T, T + 2\pi/\omega]$ when $T >> 1$. 

In \cite{pastpaper}, we derive analytical approximations for each of the above Fourier coefficients. For instance, we obtain the following result in Theorem \ref{thm1}.\begin{thm}\label{thm1} When $\abs{v}$ and $\eta$ are sufficiently small and the frequency of travel (e.g., $\omega/2\pi$) is sufficiently high, $w^{(2)}_0$ is closely approximated by a solution in the zero locus of the  nonlinear function \be\label{eqforw2} \begin{aligned} & L(w_0^{(2)}) := \frac{N - N \exp \pa{\frac{ab\,\abs{C_0}A(w^{(2)}_0)}{N + Mu}} + \frac{M(u + v)g_2(g_2 + n)w^{(2)}_0 }{ac\pa{Q_2  n - (g_2 + n)w_0^{(2)}}}}{N + (u + v)M} \ + \\ & \frac{N - N \exp \pa{\frac{ab\,\abs{C_0}A(w^{(2)}_0)}{N + Mu}} + \frac{M(u - v)g_2(g_2 + n)w^{(2)}_0 }{ac\pa{Q_2  n - (g_2 + n)w_0^{(2)}}}}{N + (u - v)M} - \frac{2g_1A(w^{(2)}_0)\pa{g_1 + n}}{ac\pa{nQ_1 - (g_1 + n)A(w^{(2)}_0)}}, \end{aligned} \ee where \be \begin{aligned} A(w^{(2)}_0) = \frac{N + Mu}{ab\,C_0u} \pa{\frac{ab\,\abs{C_0}w^{(2)}_0}{M} + \ln\pa{1 - \frac{g_2(g_2 + n)w^{(2)}_0 }{ac\,\pa{Q_2 n - (g_2 + n)w_0^{(2)}}}}} \end{aligned} \ee and 
\[C_0 = \lim_{t \to \infty} \int_0^t f(\tau) \ d\tau.\]  \end{thm}

Simultaneously, one can analytically approximate the values of $w^{(1)}_0$, $w^{(1)}_{-1}, w^{(1)}_{1}$, $w^{(2)}_{1},$ and $w^{(2)}_{-1}$ from the value of $w^{(2)}_{0}$, as described in Appendix 1 of \cite{pastpaper}. 

\begin{rem} Though the approximations in Theorem \ref{thm1} cannot be applied directly in the case of the fitted parameter values computed in Section \ref{modfitting} (for which the oscillating model solutions are in fact seen to decay towards a disease-free equilibrium (DFE), rather than approaching the periodic orbits described above), these approximations are considerably more accurate in cases where $\eta = 0$ and the endemic limiting solutions are stable \cite{pastpaper}. Moreover, even in cases where the DFE is stable, the Fourier decompositions in Equation \eqref{coeff} remain useful, and are employed in particular throughout Section \ref{mda}, to determine optimal timing for MDA across individuals in $\mathcal{M}$.

\end{rem}
\section{Model fitting for forest-village transmission}\label{modfitting}

To estimate values for the parameters in Equation \eqref{eq:humsimplify}, we first used PlotDigitizer to electronically extract monthly case data for \textit{P. vivax} in Vietnam from a time series generated in 2018 by Wangdi \textit{et al.}\cite{wangdi2018analysis} (see the top row of Figure 7 therein). Specifically, we narrowed the time interval in \cite{wangdi2018analysis} to the period between August 2011 and February 2015, during which there were four annual peaks of \textit{P. vivax} prevalence (Figure 4(a), green dots). 

When decomposing the full time-series data in \cite{wangdi2018analysis}, the authors found that only part of the fluctuation in monthly incidence could be attributed to seasonal effects of mosquito demographics, or to the general (country-wide) downward shift in malaria cases. The data that remained was highly noisy, but oscillated with a period of approximately one year (see the bottom row of Figure 7 in \cite{wangdi2018analysis}). Given that forest travel in the GMS is seasonally patterned, with peaks typically in midsummer (e.g., as shown in Rerolle \textit{\textit{et al.}} \cite{rerolle2021population}, which plotted the average number of days spent in Laos Champasak forests per month), we hypothesize that movement-associated fluctuations in prevalence could account for much of the remaining oscillations in the data in \cite{wangdi2018analysis}. 

From the solutions of our model in Equation \eqref{eq:humsimplify} (computed numerically), we then derive the \textit{incidence functions} $J_N(t), J_M(t)$, which give the \textit{P. vivax} incidence among stationary and moving populations, respectively, at time $t$. We first notice that, using Equation (2) in Anwar \textit{et al.} \cite{anwar2023optimal} (which gives the equation for the proportion of infectious individuals in a single stationary population, in the presence of superinfection), the \textit{net outflow} of primary and secondary cases in $\mathcal{N}$ is given by \[\pa{p_1^{(\mathcal{N})}(t) + p_2^{(\mathcal{N})}(t)}I_N.\] Here, $p_1^{(\mathcal{N})}(t)$ is the probability that a stationary individual has exactly one primary infection and no hypnozoites (meaning that this individual would leave compartment $I_N$ and transition to a susceptible, non-infectious state upon recovery). Similarly, $p_2^{(\mathcal{N})}(t)$ is the probability that a stationary individual has both exactly one primary infection and latent hypnozoites, meaning that this individual would transition from $I_N$ to a \textit{liver-infected} (but non-infectious) state upon recovery. 

Analogously, we can define probabilities $p_1^{(\mathcal{M})}(t)$ and $p_2^{(\mathcal{M})}(t)$ for a moving individual, such that the net outflow of primary and secondary cases in $\mathcal{M}$ becomes \[\pa{p_1^{(\mathcal{M})}(t) + p_2^{(\mathcal{M})}(t)}I_M.\] 

From Section 2.2.3 and Equation (21) in Anwar \textit{et al.} \cite{anwar2023optimal}, we derive exact expressions for $p_1^{(\mathcal{N})}(t) + p_2^{(\mathcal{N})}(t)$ and $p_1^{(\mathcal{M})}(t) + p_2^{(\mathcal{M})}(t)$. In particular, we have that 

\be \begin{aligned} & p_1^{(\mathcal{N})}(t) + p_2^{(\mathcal{N})}(t) = \frac{N - I_N(t)}{I_N(t)}\pa{\int_0^t \lambda_1(\tau) f^*(t - \tau) \, d\tau}; \\ & p_1^{(\mathcal{M})}(t) + p_2^{(\mathcal{M})}(t) = \frac{M - I_M(t)}{I_M(t)}\pa{\int_0^1 \Big[A_1(\tau)\lambda_1(\tau) + A_2(\tau) \lambda_2(\tau)\Big] f^*(t - \tau) \, d\tau}, \end{aligned} \ee
where $f^*(t)$ is the function satisfying \[f^*(t):= \frac{e^{-\gamma t} + \nu p_A(t)}{\pa{1 + \nu p_A(t)}^2}.\]
We can now define the incidence functions $J_N$ and $J_M$, as the sum of the derivatives of $I_N, I_M$, respectively, with the net outflows of cases in both populations. In particular, we write that 

\be\label{incidence_derivation} \begin{aligned} & J_N(t) = \gamma (p_1^{(\mathcal{N})}(t) + p_2^{(\mathcal{N})}(t))\,I_N + I_N'(t) = \gamma (p_1^{(\mathcal{N})}(t) + p_2^{(\mathcal{N})}(t))\, I_N + \pa{I_N(t) - N} \times \\ & \Biggr[\frac{-ab \, I_{m1} (t)}{N +  (u + v\cos \omega t) \, M} + ab\int_{0}^{t}\pa{\frac{I_{m1} (\tau)}{N +  (u + v\cos \omega \tau) \, M}\, f'(t - \tau)} d \tau\Biggr]; \\ \\ & J_M(t) = \gamma (p_1^{(\mathcal{M})}(t) + p_2^{(\mathcal{M})}(t)) \, I_M + I_M'(t) = \gamma (p_1^{(\mathcal{M})}(t) + p_2^{(\mathcal{M})}(t))\, I_M + \ \\ & \pa{I_M(t) - M} \times\Biggr[-\,ab \pa{\frac{I_{m1}(\tau)(u + v \cos \omega \tau)}{N + (u + v \cos \omega t)M}  \, + \frac{I_{m2}(\tau)((1 - u) - v\cos \omega \tau)}{((1 - u) - v\cos \omega \tau)M}} + \\ & ab\int_{0}^{t}  \pa{\frac{I_{m1}(\tau)(u + v \cos \omega \tau)}{N + (u + v \cos \omega t)M}  \, + \frac{I_{m2}(\tau)((1 - u) - v\cos \omega \tau)}{((1 - u) - v\cos \omega \tau)M}} f'(t - \tau)\, d\tau\Biggr]. \end{aligned} \ee 
 For every fixed set of parameter values, we can numerically solve for $J_N(t)$ and $J_M(t)$ using our solutions for $I_M, I_N, I_{m1},$ and $I_{m2}$, which are in turn derived using the integro-differential equations (IDE) solver constructed in Elgart \textit{et al.} \cite{pastpaper} (see Section 2.5, \cite{pastpaper}).  

We now assume values for a subset of the parameters (including $M, \omega, u$, and $v$, along with parameters considered independent of location, namely $b, c, n, \alpha, \mu, \gamma,$ and  $\nu$), as shown in Table \ref{parameters}, and vary the remaining parameters, namely $a, \tilde{g}_1, \tilde{g}_2, Q_1, Q_2,$ and $\eta$. In MATLAB, we use the \ttfamily{lsqcurvefit}\rmfamily{} and \ttfamily{nlpredci}\rmfamily{} functions to fit the modeled incidence to the time-series data (where \ttfamily{lsqcurvefit}\rmfamily{} minimizes the squared residuals, and \ttfamily{nlpredci}\rmfamily{} calculates the 95\% confidence interval associated with the fit). To prevent the fitting algorithm from concluding its search at local minima in parameter space, MATLAB's \ttfamily{multistart}\rmfamily{} function is used to compute initial values for the fitted parameters, implementing the fitting algorithm in succession from 25 scattered parameter vectors, and comparing the results.

\begin{figure}[h!] 
\centering
\includegraphics[width=\textwidth]{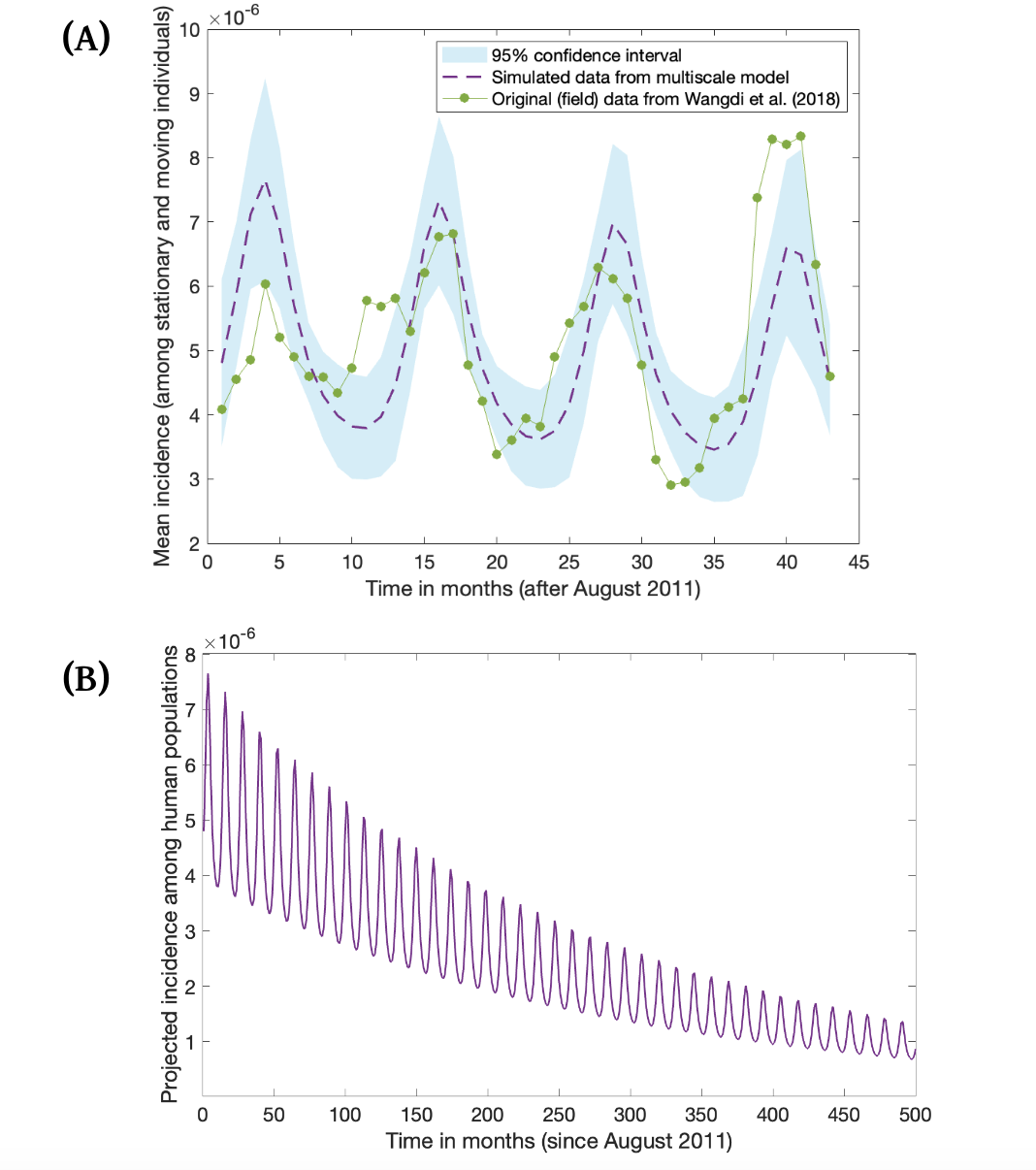}
\caption{\small (a) Fitted model (dashed purple line), original data (solid green line) and 95\% confidence intervals (shaded in cyan). The fitted model represents population-wide monthly incidence (given by $J_N(t) + J_M(t)$, where the functions $J_N, J_M$ are defined in Equation \eqref{incidence_derivation}), as a weighted average for people in $\mathcal{M}$ and $\mathcal{N}$, while the original time-series data was extracted from Wangdi \textit{\textit{et al.}} \cite{wangdi2018analysis}, and represents \textit{P. vivax} malaria cases in Vietnam between August 2011 and January 2015. The parameters of the fitted model are reported in Table \ref{parameters}. (b) Projection of fitted mean incidence across a $40$-year time-frame.} \label{ModelFitting}
\end{figure}

Figure 4(a) illustrates the original time-series from Wangdi \textit{et al.} \cite{wangdi2018analysis}, the output of the fitted multiscale model (represented by a dashed purple line), as well as the 95\% prediction interval of the fit (represented by the cyan shaded area). Though we do not perform formal identifiability analysis here, we discuss future work in this direction throughout Section \ref{discussion}.

\begin{table}[h!]\label{table1}
\centering\begin{tabular}{| m{1.5em} || m{15em} | m{7.5cm} | } \hline
  & \centering Definition & Assumed/Fitted Value \\ [0.1ex] 
 \hline\hline
  $a$ & 
  Mean daily mosquito bite rate &  $0.0935$ days$^{-1}$ (fitted)  \\ \hline
  $b$ & 
  Probability of mosquito to human transmission & 
      0.5* (assumed, \cite{smith2010quantitative})   \\ \hline
  $c$ & 
  Probability of human to mosquito transmission &  0.23* (assumed, \cite{bharti2006experimental}) \\ \hline
  $\eta$ & 
  Seasonality amplitude in Equation \eqref{eq:humsimplify:seasonality} &   0.5306* (fitted)   \\ \hline
  $\tilde{g}_1$ & 
  Patch $1$ baseline mosquito demography rate &   $0.1983$ days$^{-1}$ (fitted)   \\ \hline
  $\tilde{g}_2$ & 
  Patch $2$ baseline mosquito demography rate &  $0.0993$ days$^{-1}$ (fitted)  \\ \hline
  $n$ & 
  Mosquito sporogyny rate & $1/12$ days$^{-1}$ (assumed, \cite{gething2011modelling})  \\ \hline
  $Q_1$ & 
  Ratio of mosquitoes on Patch $1$ to the number of humans across both patches &   0.001* (fitted)   \\ \hline
  $Q_2$ & 
  Ratio of mosquitoes on Patch $2$ to the number of humans across both patches &  0.0931* (fitted)  \\ \hline
  $M$ & 
  Proportion of moving individuals &  0.1* (assumed from western Cambodia data, including \cite{kheang2021cambodia} and \cite{tripura2017submicroscopic}, which report 5.7\% and 15.6\% values, respectively, for the proportion of forest-goers in a village community; later varied)  \\ \hline
$N$ & 
Proportion of stationary individuals &  0.9* (derived from the value of $M$)   \\ \hline
$u$ & 
$A(\pi/(2\omega))$, where $A(t) = u + v \cos \omega t$ is the density of moving humans on Patch $1$ at time $t$  &  $0.5$ days$^{-1}$ (assumed) \\ \hline
$v$ & 
$A(0) - A(\pi/(2\omega))$ &  $-0.25$ days$^{-1}$ (assumed) \\ \hline
$\omega$ & 
$2\pi\, \times$ the frequency of human movement & $\frac{2\pi}{365}$ days$^{-1}$ (assumed, chosen due to the travel patterns mentioned in \cite{bannister2019forest}) \\ \hline
$\alpha$ & 
Hypnozoite activation rate & 1/332 days$^{-1}$ (assumed, \cite{white2014modelling})  \\ \hline
$\mu$ & 
Hypnozoite death rate & 1/425 days$^{-1}$ (assumed, \cite{white2014modelling})  \\ \hline
$\gamma$ & 
Blood-infection clearance rate & 1/60 days$^{-1}$ (assumed, \cite{collins2003retrospective})  \\ \hline
$\nu$ & Mean number of hypnozoites established in a bite & $5$* (assumed, \cite{white2016variation})  \\ \hline
\end{tabular}
\caption{List of literature-derived and fitted values for the parameters in the three-scale model defined by Equation \eqref{eq:humsimplify}. Of the assumed parameters, $\omega, u, v$, and $M$ are specific to the forest travel patterns in the GMS; values for the remaining non-fitted parameters were first assembled in Anwar \textit{\textit{et al.}} \cite{Anwar} (and previously used in \cite{pastpaper}). Original sources for each of the non-fitted parameters are listed in parentheses; asterisks next to a parameter value indicate the non-dimensionality of the constant.   }\label{parameters}
\end{table}

The fitted model parameters in Table \ref{parameters} suggests a considerable difference between the population sizes and lifespans of mosquitoes in the forest and village patches: in particular, the Patch 1 mosquito population (captured by the density parameter $Q_1 = 0.001$) is over $90$ times smaller than the population of Patch $2$ (captured by the density $Q_2 = 0.0931$), with an estimated lifespan of only $1/\tilde{g}_1 \approx 5.04$ days for Patch 1 mosquitoes, while Patch 2 mosquitoes are estimated to survive for an average of approximately $1/\tilde{g}_2 \approx 10.07$ days. This suggests effective mosquito control in villages, indicating that insecticidal nets and indoor residual spraying have limited both mosquito-to-human contact and mosquito presence. 

Moreover, we observe that the fitted incidence in Figure 4(a) decreases consistently over the four-year interval shown. Projecting the solutions for the mean incidence over several more decades implies that the system approaches the DFE, with incidence decreasing seven-fold over the forty-year period following August 2011 (Figure 4(b)). This sustained ebb in caseload aligns with the broader trend towards malaria eradication noted in Section \ref{Intro}.

\section{Modelling non-pharmaceutical and pharmaceutical malaria interventions}

Having derived reasonable parameters for the multiscale model in the case of a large population in the GMS, we now apply the model to investigate best-possible \textit{P. vivax} interventions for forest-adjacent areas in the region. 

We begin with the MDA of antimalarials, which has been shown to consistently reduce malaria incidence in small communities across Vietnam, Cambodia, and Laos \cite{von2019impact}. Since a traveling individual's exposure to \textit{P. vivax} may vary considerably over the course of a season, understanding the relationship between MDA timing and the effect of a dose may prove significant in improving MDA efficacy. To the best of our knowledge, however, no mathematical modelling work has considered the quantitative problem of optimizing MDA times for traveling forest-goers specifically, so as to maximally reduce the risk of infection. 

As such, we first consider a one-dose focal MDA scheme for the moving group in our metapopulation model, which involves both hypnozoite- and merozoite-targeting medications (e.g., an artemisinin-based combination therapy). We then derive the best-possible MDA time given a collection of simplifying assumptions (requiring, in particular, that the efficacy of the antimalarial drug be sufficiently high, and that the period of travel be sufficiently short). 
\subsection{Optimal time of MDA: analytical results}\label{mda}

We assume that every individual in $\mathcal{M}$ receives one round of MDA at a fixed population-wide time $s$, where $0 < s < t$. The efficacy of this MDA round is associated with $p_{\textnormal{blood}}$ and $p_{\textnormal{rad}}$, the probability that a blood-stage infection undergoes clearance upon administration and the probability that a given hypnozoite dies following antimalarial treatment, respectively. Both clearance and hypnozoite death are assumed to occur instantaneously. 

We hypothesize that the effect of antimalarial distribution will depend on the proportion of moving individuals on each patch at time $s$ (since these quantities determine the mean population FORI at any given moment), and will therefore vary cyclically over each movement period (of length $2\pi/\omega$). In this section, we seek to find the best-possible value of $s$ within a single movement period, so as to minimize the mean prevalence of erythrocytic malaria among forest-going individuals over a specified interval of time. 

We first formalize the optimization problem. Using Equation (31) in Mehra \textit{et al.} \cite{Mehra}, we determine that, in the presence of MDA, the PGF capturing within-host dynamics for an individual with periodic movement function $B(t)$ can be written as

\begin{equation}\label{eq:pgfmovingmda}
\begin{aligned} 
& G_{M}(\tilde{z}_H, \tilde{z}_A, \tilde{z}_C, \tilde{z}_D, \tilde{z}_P, \tilde{z}_{PC}) := \mathbb{E} [\prod_{j \in J} \tilde{z}_{j}^{N_j(t)}] = \\ &  
\underbrace{\exp \pa{\int_0^s  \Lambda(\tau)  \pa{\frac{\pa{\tilde{z}_{P}e^{-\gamma(t - \tau)} + (1 - (1-p_{blood})e^{-\gamma(t - \tau)})\tilde{z}_{PC}}}{1 + v\pa{1 - \sum_{j \in J_h} \tilde{z}_{j} \cdot p^r_{j}(t - \tau, s - \tau)}} - 1} \hspace{3pt} d\tau}}_{\textnormal{before MDA}} \times  \\ & \underbrace{\exp \pa{\int_s^t \Lambda(\tau) \pa{\frac{\pa{\tilde{z}_{P}e^{-\gamma(t - \tau)} + (1 - e^{-\gamma(t - \tau)})\tilde{z}_{PC}}}{1 + v\pa{1 - \sum_{j \in J_h} \tilde{z}_{j} \cdot p_{j}(t - \tau)}} - 1} \hspace{3pt} d\tau}}_{\textnormal{after MDA}},
\end{aligned}
\end{equation}
where for brevity we introduce the shorthand \[\Lambda(\tau):= \lambda_{1}(\tau)B(\tau) + \lambda_{2}(\tau)\pa{1 - B(\tau)},\] and $J_h$, as in Section \ref{withinhost}, represents the set $\{H, A, C, D\}$ of state compartments for a single hypnozoite. Moreover, the probabilities $p_H^r(t, s), p_A^r(t, s), p_C^r(t, s), p_D^r(t, s)$ are analogues of the associated probabilities $p_H(t), p_A(t), p_C(t)$, and $p_D(t)$ after MDA at time $s$ (as derived in Equations (23)-(27) of Mehra \textit{et al.} \cite{Mehra}).

Using Equation (68) in Mehra \textit{et al.} \cite{Mehra}, it follows that the probability of an individual with movement function $B(t)$ to be infectious at time $t$ is \be \begin{aligned}\label{pIts} p_I(t,s) = 1 - P\pa{\mathcal{N}_A(t) + \mathcal{N}_P(t) = 0} = 1 - G_M\pa{1, 0, 1, 1, 0, 1} = \\ 1 - \exp\pa{\int_s^t \Lambda(\tau) \, f(t - \tau)\ d\tau}\times \exp\pa{\int_0^s \Lambda(\tau) \, r(t - \tau, s - \tau) \ d \tau},\end{aligned} \ee
where $f(t)$ is defined as in Equation \eqref{f_Eq} and \be \begin{aligned} \label{req} r(t, s) = \frac{1 - (1-p_{\textnormal{blood}})e^{-\gamma t}}{1 + \nu((1 - p_{\textnormal{blood}})e^{-\gamma(t-s)}p_A(s) + (1 - p_{\textnormal{rad}})\pa{p_A(t) - e^{-\gamma(t-s)}p_A(s)})} - 1 \\ = \frac{1 - (1-p_{\textnormal{blood}})e^{-\gamma t}}{1 + \nu((1 - p_{\textnormal{rad}})p_A(t) + (p_{\textnormal{rad}} - p_{\textnormal{blood}})e^{-\gamma(t-s)}p_A(s))} - 1. \end{aligned} \ee  

Moreover, for any given $t$ and $s < t$, we can show that $r(t, s) < 0$, in that \be \begin{aligned} p_{A}(t) - e^{-\gamma(t-s)} p_A(s) = \frac{\alpha}{\alpha + \mu - \gamma}\pa{e^{-\gamma t} - e^{-(\alpha + \mu) t} - e^{-\gamma t} + e^{-(\alpha + \mu - \gamma) s - \gamma t}} = \\ \frac{\alpha}{\alpha + \mu - \gamma}\,e^{-(\alpha + \mu) t}\pa{e^{(\alpha + \mu - \gamma) (t-s)} - 1} > 0, \end{aligned} \ee (where we note that $e^{(\alpha + \mu - \gamma) (t-s)} - 1 < 0$ if and only if $\alpha + \mu - \gamma < 0$).
Moreover, again using the relation $p_A(t) > e^{-\gamma(t - s)} p_A(s)$, we find that \be \begin{aligned} \abs{r(t,s)} = 1 - \frac{1 - (1-p_{\textnormal{blood}})e^{-\gamma t}}{1 + \nu((1 - p_{\textnormal{rad}})p_A(t) + (p_{\textnormal{rad}} - p_{\textnormal{blood}})e^{-\gamma(t-s)}p_A(s))} < \\ 1 - \frac{1 - e^{-\gamma t}}{1 + \nu((1 - p_{\textnormal{rad}})p_A(t) + (p_{\textnormal{rad}} - p_{\textnormal{blood}})p_A(t))} = \\ 1 - \frac{1 - e^{- \gamma t}}{1 + \nu (1 - p_{\textnormal{blood}}) p_{A}(t)} < 1 - \frac{1 - e^{- \gamma t}}{1 + \nu p_{A}(t)} = |f(t)|, \end{aligned} \ee 
so \[|r(t, s)| < |f(t)| < 1.\]

As such, we can apply Theorem \ref{thmEq} with the $f(t-\tau)$ term in the integral replaced by the piecewise-defined term \[ \begin{cases} 
      r(t - \tau, s - \tau) & \tau < s \\
      f(t - \tau) & s\leq \tau \leq t.
   \end{cases}
\]
This yields that the proportion of infectious moving individuals at time $t$ following MDA at time $s$ is \be \begin{aligned} \tilde{I}_M(t, s) = 1 - \exp\pa{\int_s^t \pa{\lambda_{1}(\tau)A_1(\tau) + \lambda_{2}(\tau)A_2(\tau)} \, f(t - \tau)\ d\tau}\times \\ \exp\pa{\int_0^s \pa{\lambda_{1}(\tau)A_1(\tau) + \lambda_{2}(\tau)A_2(\tau)} \, r(t - \tau, s - \tau) \ d \tau}.\end{aligned} \ee

Making the assumption that $s >> 1$ (i.e., as $s$ tends to infinity), we note that $\tilde{I}_M(t, s)$ is limit-periodic in $t$ for a fixed $s$. In particular, when $s >> 1$ is held constant, \[\lim_{t-s\to\infty} \exp\pa{\int_0^s \pa{\lambda_{1}(\tau)A_1(\tau) + \lambda_{2}(\tau)A_2(\tau)} r(t - \tau, s - \tau) \ d \tau} = 1,\] since $\lim_{k \to \infty} r(k, l) = 0$ for all fixed $l > 0$. As such, \[\lim_{t - s \to \infty} \tilde{I}_M(t, s) = \lim_{t - s \to \infty} \Biggr[1 - \exp\pa{\int_s^t \pa{\lambda_{1}(\tau)A_1(\tau) + \lambda_{2}(\tau)A_2(\tau)} f(t - \tau)\ d\tau}\Biggr],\] which is a limit-periodic function (by the convergence of $f$ to $0$ as $t$ tends to infinity), approaching a period of $\frac{2\pi}{\omega}$. To determine an ``optimized" value for $s$, it thus suffices to minimize the \textit{aggregate prevalence} function  
\be\label{optimization_problem} AP(s) = \int_0^{\frac{2 \pi}{\omega}} \tilde{I}_M(s + k, s) \ dk\ee
over the interval $s \in [T, T + 2\pi/\omega]$, where $T >> 1$ is measured in days. This quantity is equivalent to an aggregate erythrocytic malaria prevalence in the first annual cycle (of length $2\pi/\omega$ days) following antimalarial administration.

We obtain the solution to the optimization problem using Theorem \ref{thmMDA}:
\begin{thm}\label{thmMDA}
When $\omega$ is sufficiently large and $p_{\textnormal{blood}}, p_{\textnormal{rad}}$ are sufficiently close to $1$, the optimal annual time of MDA relative to the start-of-movement-period time $T$ is the unique solution of 
\be\label{soleq} \begin{aligned} & \tan \,(\omega (s-T))  =  -\frac{Mu(w^{(1)}_{-1} + w^{(1)}_{1}) + (N+Mu)(w^{(2)}_{-1} + w^{(2)}_{1})}
{Mu (w^{(1)}_{1} - w^{(1)}_{-1})\textbf{i}  + (N+Mu)(w^{(2)}_{1} - w^{(2)}_{-1})\textbf{i} }
\end{aligned}
\ee that satisfies the constraints 
\be\label{constraints} \begin{aligned}
-\pa{Mu(w^{(1)}_{-1} + w^{(1)}_{1}) + (N+Mu)(w^{(2)}_{-1} + w^{(2)}_{1})} \sin \,(\omega (s-T)) + \\ \pa{Mu (w^{(1)}_{1} - w^{(1)}_{-1})\textbf{i}  - (N+Mu)(w^{(2)}_{1} - w^{(2)}_{-1})\textbf{i} } \cos \,(\omega (s-T)) < 0
\hspace{10pt} 
\end{aligned}
\ee and $s - T \in [0, 2\pi/\omega),$ \\
\end{thm}
where the coefficients $w^{(1)}_{-1}, w^{(1)}_{1}, w^{(2)}_{-1}, w^{(2)}_{1}$ are as in Equation \eqref{coeff}.
Theorem \ref{thmMDA} is proved in Appendix II below. 

\begin{rem}
We notice here that the results of Theorem \ref{thmMDA} are most applicable when $I_{m1}(t), I_{m2}(t)$ are arbitrarily close to a stable limit cycle. As mentioned above, in the case of the fitted model parameters in Table \ref{parameters}, the solutions for $I_{m1}(t)$ and $I_{m2}(t)$ in fact approach the DFE, limiting the accuracy of the analytical approximations we compute in Section \ref{numericalResults} below. To demonstrate the validity of the approximations in the case where the solutions tend to periodic ones, we also consider a case where the value for the per-capita bite rate $a$ is higher, in Figure 7 of Section \ref{numericalResults}.
\end{rem}

\begin{figure}[h!]
\centering
\includegraphics[width=0.8\textwidth]{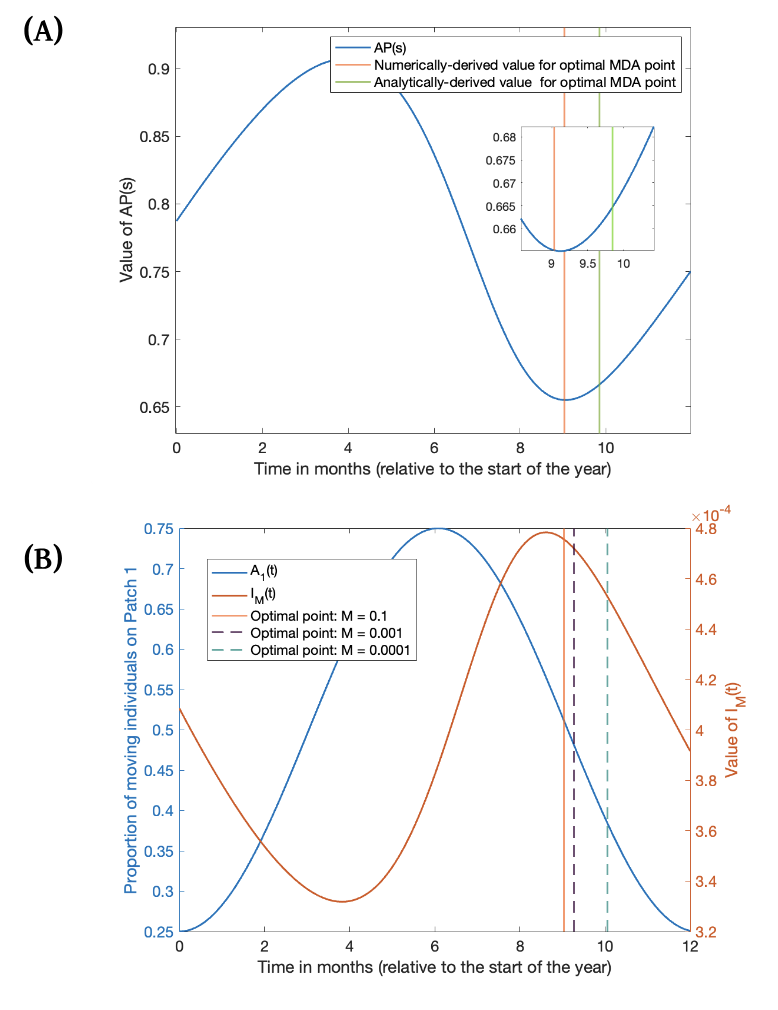}
      \caption{\small (a) Illustrates the value of $AP(s)$ from Equation \eqref{optimization_problem}, where $s$ ranges on the interval $[9855 \textnormal{ days}, 10220\textnormal{ days}]$ (equivalent to $[0 \, \textnormal{months}, 12 \, \textnormal{months}]$ after the start of the 27th year of the simulation), given parameters in Table \ref{parameters}. The numerical solution ($\approx 9.033$ months) for the minimum $s$-value is marked with a vertical orange line, while the analytical solution ($\approx 9.844$ months, derived using Equation \ref{soleq}), is marked with a vertical green line. The inner box contains a zoomed-in version of a section of the $AP(s)$ curve, illustrating the difference between the two obtained values for the optimal MDA timing. (b) Illustrates the proportion of moving individuals  on Patch 1 (blue axis) and the proportion of infectious moving individuals (red axis) as a function of time, with numerically-optimal MDA times marked for various values of $M$. }\label{fig5}
     \end{figure}
\subsection{Optimal time of MDA: numerical results}\label{numericalResults}

In this section, we consider the approximate ``optimal point" of MDA given by Theorem \ref{thmMDA} in relation to the periodic patterns of migration and disease prevalence (as modeled by Equations \eqref{A1A2} and \eqref{densityofmovinginf}, respectively). In particular, we find that, for the specific parameter values in Table \ref{parameters}, the optimal point of MDA distribution (when derived numerically) is reached when approximately half of forest-going individuals are still found on Patch 1, which aligns approximately with the peak transmission period for moving individuals.      

Setting the parameters equal to the values obtained after model fitting (see Table \ref{parameters}), we  minimized the value of $AP(s)$ in Equation \eqref{optimization_problem}. In particular, we numerically simulate the model in Equation \eqref{eq:humsimplify} (with the mosquito demography rates, $g_1(t)$ and $g_2(t)$, as in Equation \eqref{eq:humsimplify:seasonality}) to determine values for $I_{m1}(t)$ and $I_{m2}(t)$. Beginning with the 27th year of the simulation, we then evaluate $AP(s)$ (with $p_{\textnormal{blood}}, p_{\textnormal{rad}}$ set at \textit{baseline} values of $0.9$), such that $s$ varies over the period $[9855 \textnormal{ days}, 10220 \textnormal{ days}]$. In this case, the minimum value for $AP(s)$ falls approximately $9.033$ months after the beginning of the year ($s \approx 10130$ days, Figure 5(a)). 

Evaluating solutions to Equation \eqref{soleq} and selecting the value of $s$ that satisfied the constraints in Equation \eqref{constraints} yielded an analytical minimum at approximately $9.844$ months after the beginning of the year ($s \approx 10154$), resulting in an error of around three-and-a-half weeks (Figure 5(a)). This error is heightened by the decrease in malaria cases over time, which distorts the shape of the yearly prevalence curve from a sinusoidal one (Figure 5(b)). 

For the baseline parameter values, the optimal MDA time coincides with the point at which slightly more than 50\% of the moving individuals are found on Patch 1, such that there the net outflow from Patch 1 to Patch 2 is nearing a maximum (Figure 5(b), blue curve). This best-possible value for $s$ suggests the strategy of applying MDA immediately prior to the peak village departure period, determined across all migrating people in a population. This optimal time also falls approximately two weeks after the prevalence of infectious malaria peaks in the moving population (Figure 5(b), red curve).

Interestingly, the numerical simulation indicates that the peak in infectiousness (and therefore the optimal MDA time) occurs when the group-level FORI of individuals in $\mathcal{M}$ is average (with individuals divided evenly between the two patches), as opposed to at its highest. This phenomenon is likely due to the dependence of $I_M(t)$ on the convolution integrals over $\lambda_1(t)$ and $ \lambda_2(t)$ in Equation (\ref{eq:humsimplify}), rather than on the instantaneous FORI itself.

\begin{figure}[h]
     \centering
\includegraphics[width=0.8\textwidth]{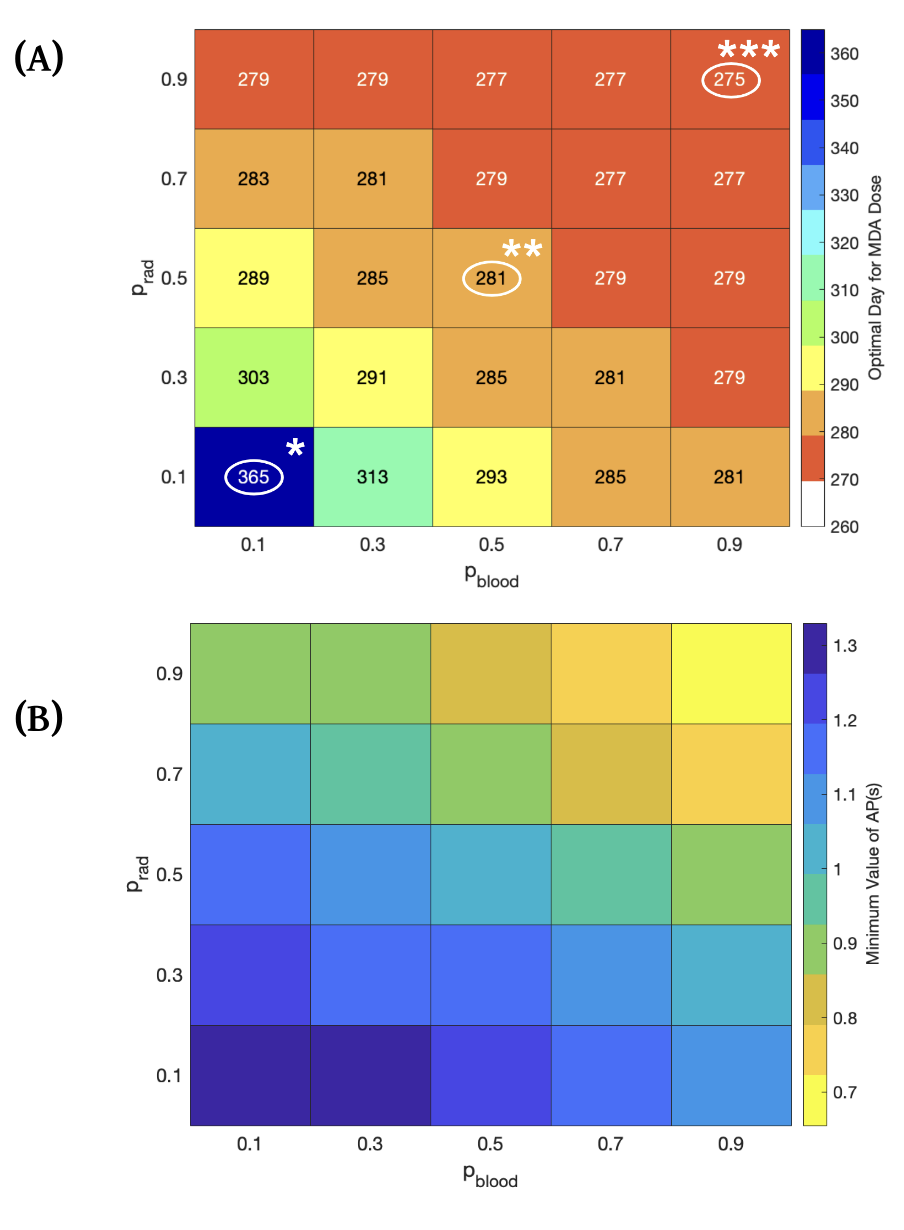}
     \caption{\small (a) Heatmap illustrating the optimal time of MDA administration (from Equation \eqref{soleq}) in days after the beginning of the $27$th simulated year. Here, we vary $p_{\textnormal{blood}},  p_{\textnormal{rad}} \in \{0.1, 0.3, 0.5, 0.7, 0.9\}$, and fix $M = 0.1$. Values from Table \ref{table2} are indicated with asterisks. (b) Heatmap showing the mean aggregate \textit{P. vivax} prevalence among moving individuals for values of $p_{\textnormal{blood}},  p_{\textnormal{rad}} \in \{0.1, 0.3, 0.5, 0.7, 0.9\},$ where $M = 0.1$, given one MDA round at time $s^*$ (computed using Equation \eqref{optimization_problem}). The prevalence is slightly more sensitive to changes in $p_{\textnormal{rad}}$, aligning with experimental consensus that the hypnozoite reservoir accounts for a considerable proportion of the total \textit{P. vivax} burden, \cite{white2014modelling}. }
     \end{figure}

\vspace{1 cm}

We now consider the dependence of the numerically computed optimal MDA time, denoted $s^*$, on the model parameters. In particular,  we observe that decreases in either the efficacy of the antimalarial medication or the proportion $M$ of traveling individuals delays the optimal MDA timing further from the start of the migration period (Figure 6(a), Table \ref{table2}). The dependence of $s^*$ on  $p_{\textnormal{blood}}$ and $ p_{\textnormal{rad}}$ falls as the antimalarial efficacy improves, corroborating the $p_{\textnormal{blood}}, p_{\textnormal{rad}}$-independent result in Theorem \ref{thmMDA} (Figure 6(a), columns 3-5 of Table \ref{table2}). As $M$ grows very small, the dependence of the optimal MDA timing on the value of $M$ increases dramatically. 

Furthermore, though $s^*$ depends largely only on $M, N,$ and $u$ when $p_{\textnormal{blood}}, p_{\textnormal{rad}}$ tend to $1$, our numerical simulations suggest that the lowest-possible aggregate prevalence among moving individuals (equivalent to the value of $AP(s^*))$ is strongly sensitive to changes in  $p_{\textnormal{blood}}$ and $p_{\textnormal{rad}}$ (Figure 6(b)). 

\vspace{10pt}

\begin{table}[h!]
\centering\begin{tabular}{| m{6em} || m{6em} | m{6em} | m{6em} | m{6em} |} \hline
  {\begin{center} Value of $M$ \end{center}} & {\begin{center} $p_{\textnormal{blood}} = 0.1$ \\ $p_{\textnormal{rad}} = 0.1$ \end{center}} & {\begin{center} $p_{\textnormal{blood}} = 0.5$ \\ $p_{\textnormal{rad}} = 0.5$ \end{center}} & {\begin{center} $p_{\textnormal{blood}} = 0.9$ \\ $p_{\textnormal{rad}} = 0.9$ \end{center}} & {\begin{center} $p_{\textnormal{blood}} = 0.99$ \\ $p_{\textnormal{rad}} = 0.99$ \end{center}}\\ [0.1ex] 
 \hline\hline
 \centering 0.0001 & \centering 365 & \centering 327 & \centering 315 & \hspace{2em} 311\\ [0.1ex] 
\hline
 \centering 0.001 & \centering 365 & \centering 285 & \centering 281 & \hspace{2em} 279\\ [0.1ex] 
\hline
\centering  0.1 & \centering 365* & \centering 281** & \centering 275*** & \hspace{2em} 275 \\ [0.1ex] 
\hline
\centering  0.9 & \centering 365 & \centering 281 & \centering 275 & \hspace{2em} 275\\ [0.1ex] 
\hline
\end{tabular}
\caption{\small The impact of varying $M$, the proportion of moving individuals in the human population, on the numerically-optimal timing of MDA for various values of $p_{\textnormal{blood}}, p_{\textnormal{rad}}$. As $M$ decreases, the optimal timing is increasingly delayed; similarly, as the values of $p_{\textnormal{blood}}, p_{\textnormal{rad}}$ fall, the best-possible day of antimalarial administration increases relative to the beginning of the year. The dependence of the optimal MDA timing on the efficacy of the antimalarials is negligible for sufficiently-high $p_{\textnormal{blood}}, p_{\textnormal{rad}}$ (values labeled with asterisks are also shown on Figure 6(a).)}\label{table2}
\end{table}

To demonstrate the accuracy of our analytical approximations in the case where the solutions tend towards periodic ones, we set $a = 0.2$ days$^{-1}$ (for which the endemic limit cycle is asymptotically stable), and compare the numerically- and analytically-obtained values for the optimal MDA timing (Figure \ref{EndemicApproximations}). In this case, the error incurred by the analytical approximations falls below eight days ($< 3\%$). 

\begin{figure}
 \centering
         \includegraphics[width=0.7\textwidth]{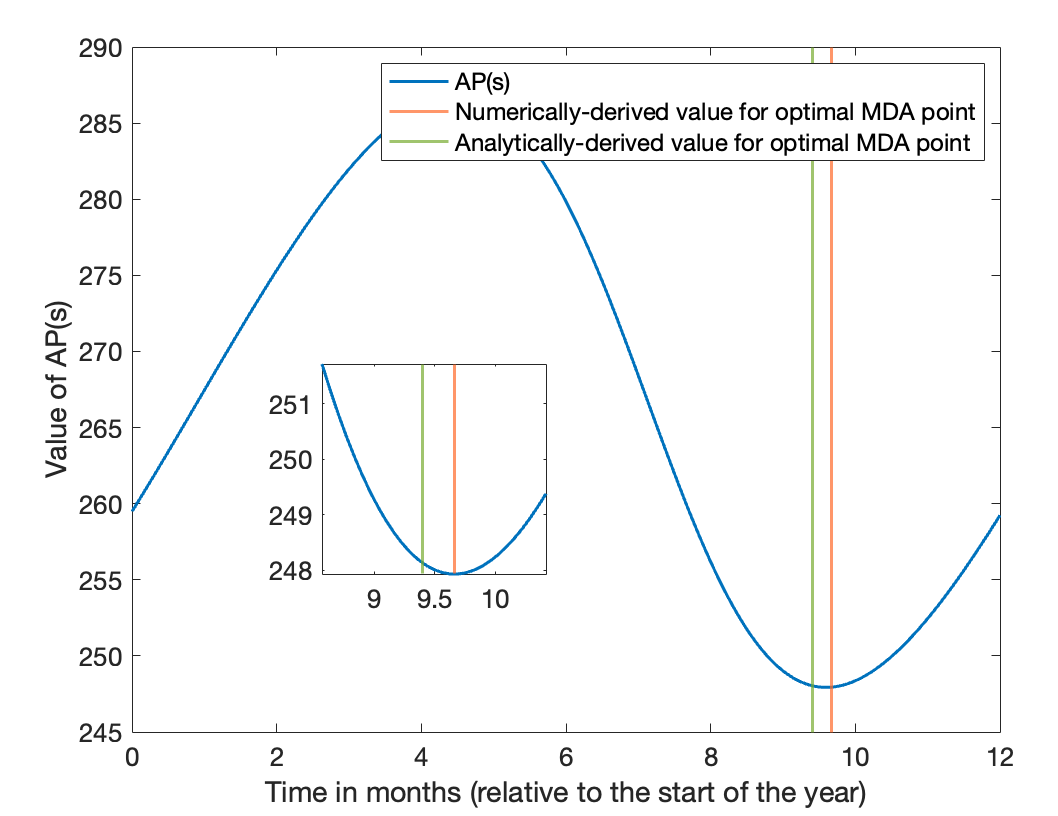}
         \caption{Comparison between the numerically and analytically-derived values for the optimal MDA time when $a = 0.2$ days$^{-1}$, computed using Equations \eqref{optimization_problem} and \eqref{soleq}, respectively.}\label{EndemicApproximations}
\end{figure}
\subsection{\textit{Anopheles}-based and non-pharmaceutical interventions}

In this section, we consider the impact of several mosquito-centered and non-pharmaceutical interventions (NPIs) on \textit{P. vivax} prevalence among stationary and moving individuals. We compare the following population-wide tools for malaria elimination in our numerical simulations:

\begin{enumerate}[label=(\roman*)]
    \item\label{inter1} \textit{Interventions targeting mosquito birth and death rates, the size of vector populations, and mosquito feeding behaviors.} This includes long-lasting insecticide-treated nets (LLINs), which impact the parameters $a, g_1, Q_1$; indoor residual spraying (IRS), which impacts $g_1, Q_1$ (we assume that IRS does not impact the mean bite rate per mosquito, as it kills rather than repels affected mosquitoes); spatial repellents (SPR), which impact $g_1, g_2, Q_1, Q_2, a$ (as SPR can be applied both in outdoor and indoor areas); and source reduction (SOR), which impacts $Q_1, Q_2$. Other methods falling into this category include changes to housing construction (e.g., elevating homes above ground level), \cite{charlwood2003raised}, which decreases $a$, and sterile insect techniques (SIC), which decrease  $Q_1, Q_2, g_1$ and $g_2$ \cite{klassen2009introduction}.
    \item\label{inter2} \textit{Interventions targeting mosquito sporogyny.} In recent years, several chemical interventions have been proposed to decrease the rate of  sporogyny ($n$), thereby forcing bottlenecks for \textit{P. vivax} parasites between the gametocyte and oocyst stages of development within the vector \cite{kamiya2022targeting}.  
\end{enumerate}
\vspace{0.4cm}
Our results indicate the relative importance of SPR, SOR, and housing construction changes for reducing the probability of infection among stationary and moving populations (Figure \ref{sensanalysis}). In the case of moving populations, the limit (periodic) prevalence $I^{*}_M(t)$ (defined in Equation \eqref{coeff}) is almost entirely independent of the parameters $g_1$ and $Q_1$, but strongly dependent on $g_2$ and $Q_2$ (Figures 8(a), 8(c)). 
 The limiting prevalence among stationary individuals, in contrast, exhibits much greater sensitivity to $Q_1$, but is nonetheless much more strongly dependent on $g_2$ than on $g_1$ (Figures 8(b),  8(d)). 
 
 Simultaneously, both $I_M^*$ and $I_N^*$ are considerably more dependent on the bite rate $a$ than on the sporogyny rate $n$ (Figures 8(e), 8(f)). The sensitivity of $I_M$ on the mosquito bite rate $a$ is especially striking: For instance, a bite rate of $a \approx 0.2$ days$^{-1}$ leads to over 70\% of the moving population becoming infectious at the limit cycle (since we set $M = 0.1$), compared to less than 0.01\% of the stationary population (Figures 8(e), 8(f)).
 
 In cases where a population contains a significant proportion of forest-going individuals, our results suggest a need to employ outdoor interventions (e.g., SOR), as well as interventions specifically targeting the mosquito bite rate (such as SPR). Even in populations with an overwhelming majority of stationary individuals (implying that epidemics tend to arise and spread on Patch 1, rather than Patch 2), SOR and SPR should be prioritized alongside LLINs and IRS, as the impact of changes in $g_2$ on $I_N$ far exceed the impact of changes in $g_1$.

\begin{figure}[h!]
\includegraphics[width=\textwidth]{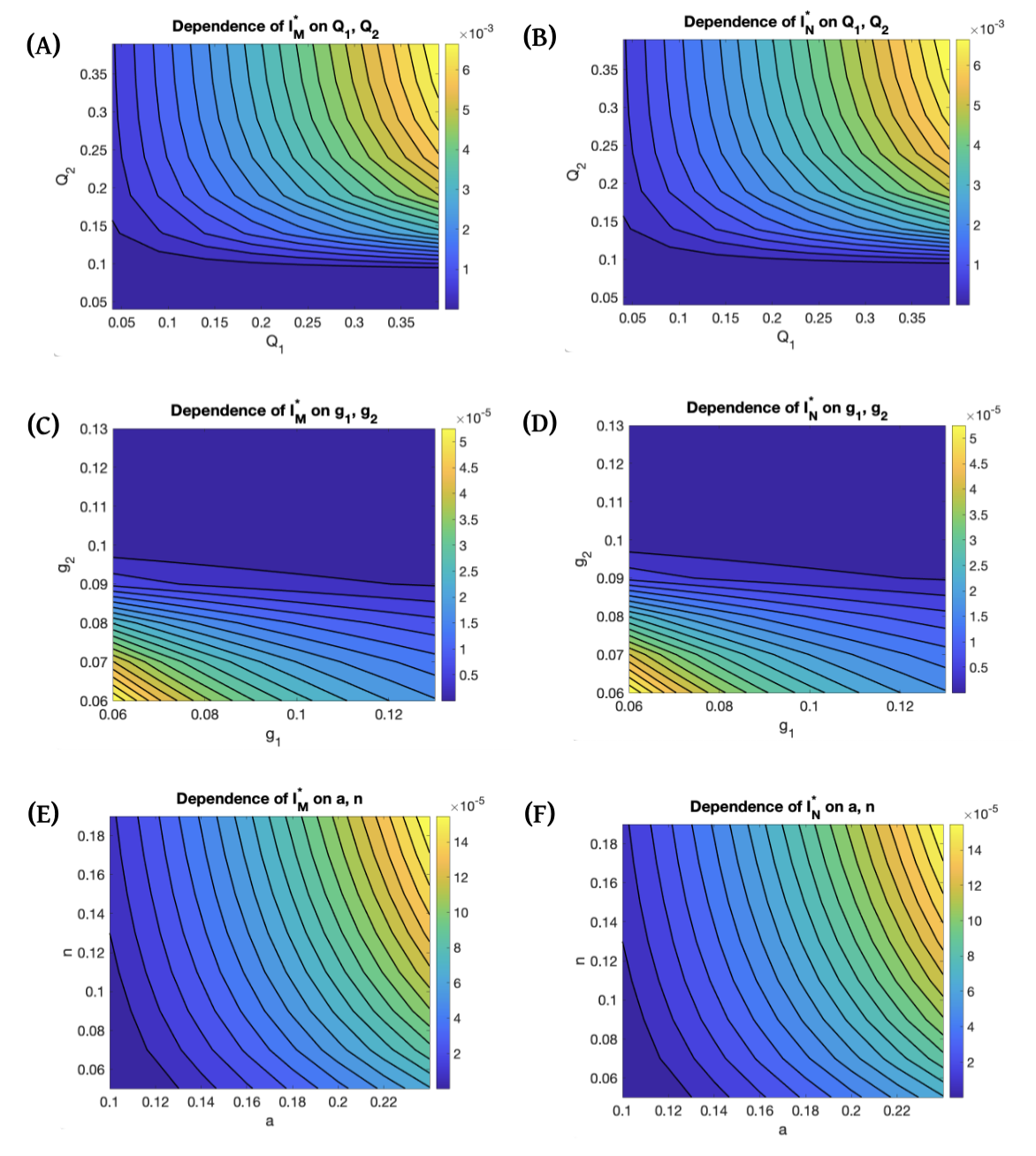}
  \caption{\small Two-dimensional contour plots that illustrate how \textit{P. vivax} prevalence in moving and stationary individuals ($I_M^*$ and $I_N^*$, respectively, simulated using Equation \ref{eq:humsimplify}) reacts to changes in [(a) and (b)] vector population size, modeled by parameters $Q_1, Q_2$, [(c) and (d)] vector demographics, modeled by parameters $g_1, g_2$, and [(e) and (f)] the bite rate and sporogony rate among all mosquitoes, modeled by parameters $a$ and $n$. Here, $Q_1, g_1$, and $a$ are plotted on the $x$-axes, while $Q_2, g_2$, and $n$ are plotted on the $y$-axes; warmer colors indicate higher levels of prevalence, while the constant dark purple color indicates complete or near-complete disease eradication. }\label{sensanalysis}
\end{figure}

\section{Discussion}\label{discussion}
In this paper, we extend our previous work on mathematical metapopulation models for \textit{P. vivax} in \cite{pastpaper} to consider the interplay between spatial structure, the hypnozoite accrual characteristic of \textit{P. vivax}, and the pharmaceutical and non-pharmaceutical interventions used to combat \textit{P. vivax} malaria. We focus on endemic \textit{P. vivax} in forest-fringe areas of Southeast Asia (e.g., Cambodia, Laos, and the Central Highlands region of Vietnam), where villages often contain a sizable “forest-going” population that works frequently in higher-risk forested areas \cite{jongdeepaisal2021acceptability}.

In particular, we fitted an expanded version of the three-scale metapopulation model developed in \cite{pastpaper} to data retrieved from Wangdi \textit{et al.}  \cite{wangdi2018analysis}, which reflects \textit{P. vivax} caseloads reported by the Vietnam National Institute of Malariology, Parasitology, and Entomology between 2011 and 2015. To better replicate real-life epidemiological conditions, we introduced seasonal variation in mosquito demography rates (Equation \eqref{eq:humsimplify:seasonality}) to the integro-differential equations system (Equation \eqref{eq:humsimplify}). The resulting model was then fitted to the data using a least-squares approach, with the addition of seasonality used to distinguish between the seasonal and movement-induced changes in the \textit{P. vivax} prevalence over time. The fitted parameters included a lower mean daily mosquito bite rate ($a \approx 0.0935$ days$^{-1}$) than that concluded in earlier works (e.g., Garrett-Jones \textit{et al.} \cite{garrett1964human}, determined a value of $0.21$ days$^{-1}$ in 1964); this suggests that the increased use of LLINs and other interventions in Vietnam may have lowered the mosquito proximity to human villages over the last several decades (Table \ref{parameters}). 

Additionally, and for similar reasons, the model fitting suggests a low proportion of mosquitoes on Patches 1 and 2 (relative to the total human population), with $Q_1$ and $Q_2$ equal to approximately $0.001$ and $0.0931$, respectively (Table \ref{parameters}). The significant difference between the value for $Q_2$ and that for $Q_1$ corroborates the observation that most mosquitoes are now concentrated within forests, as opposed to LLIN-heavy village areas, reinforcing experimental data \cite{dysoley2008changing}. As a result of these limitations on malaria transmission, the fitted solutions decayed rapidly towards the disease-free equilibrium, with model projections indicating a seven-fold decrease in mean caseload over the course of four decades. Finally, the simulated model supports the conclusion that forest migration has an annual period; observations in \cite{rerolle2021population} stress differences in migration patterning between the rainy and dry seasons of the year.

Assuming a single mass-drug administered dose of a hypnozoite-targeting antiparasite (e.g., primaquine, in combination with a broader artemisinin-based regimen to eliminate erythrocytic parasites), we used an analytical argument to determine the optimal dose timing within a single year, given Fourier approximations for the prevalence among the mosquito populations (Figure \ref{fig5}). (In the absence of seasonality, an approximate solution for the corresponding Fourier coefficients is found in terms of the model parameters in \cite{pastpaper}.) We found that the efficacy of a fixed dose oscillated over the course of the year, due to interactions between the periodic human movement and changes in the history-dependent force of reinfection. Our analytical computations for the optimal MDA time (which were developed in the case of a model system that tends towards a stable endemic limit cycle, as opposed to the case when prevalence decreases over time) incurred approximately 8.5\% error (with this error falling below $3\%$ when an endemic region of parameter space was tested instead). 

Moreover, our numerical simulations implied that, given the parameter values in Table \ref{parameters}, the best-possible MDA timing is reached around the time that the forest-goer disease prevalence is approximately highest, which falls two days before the period at which, on average, a forest-going individual migrates to Patch 2 (Figure 5(b)). The observation that maximum caseload among moving individuals aligns with the (average) inception of seasonal forest travel is supported by experimental findings: both events have been seen to occur close to the beginning of the rainy season in Cambodia and Laos \cite{kerkhof2016geographical, rerolle2021population}. In a scenario where individual forest-going teams have travel patterns aligning with the community-wide mean (e.g., with trips peaking in July or August, and tapering off in the following months, as in Rerolle \textit{et al.} \cite{rerolle2021population}), it may be that the best-possible MDA strategy is the administration of antimalarials to each team soon before departure. (This claim may be amenable to future work, particularly through an agent-based approach or other computational study.)

The optimal MDA timing is modulated strongly by the population demography (i.e., the proportion of moving and stationary individuals) and partially by the efficiency of the treatment with respect to latent and erythrocytic \textit{Plasmodium} parasites (represented by the parameters $p_{\textnormal{blood}}, \, p_{\textnormal{rad}}$ above). For instance, as the proportion of moving individuals increases, the optimal timing approaches the average beginning of forest travel (which, as mentioned above, tends to align with the beginning of the rainy season). 

We have also used the model (in Equation \eqref{eq:humsimplify}) to establish the relative efficacies of non-pharmaceutical interventions such as LLINs, SOR, SPR, and IRS (Figure \ref{sensanalysis}). We determine that, particularly in the case of a considerable moving population, LLINs may exceed IRS in efficacy, since the former specifically targets the per-mosquito bite rate, which in turn has considerable impact on the limit-cycle values for the proportion of infectious forest-going individuals. 

Owing to the mathematical complexity of the model in Equation \eqref{eq:humsimplify}, we constrain the spatial structure to only two disjoint sites---a single forest and village with varying FORIs---unlike, for instance, the agent-based model in Gerardin \textit{et al.}  \cite{gerardin2018impact}, which considered MDA interventions for \textit{P. falciparum} epidemics in two villages and a forested area. Additionally, we consider mean movement dynamics as opposed to adopting an agent-based approach, which makes the model analytically tractable, but tends to exclude the impact of individual-level stochastic behavior. In Southeast Asia, populations enter the forest seasonally, but tend to migrate in multiple teams, traveling at slightly different times across a period of several months \cite{rerolle2021population}. In a third limitation of the results presented here, our MDA optimization considers only a single round of radical cure, which (while being the preferred approach for tafenoquine, as discussed in Lacerda \textit{et al.} \cite{lacerda2019single}) is simpler than the work Anwar \textit{et al.} \cite{anwar2023optimal}, which considers an arbitrary number of MDA rounds. 

In future work, we hope to address some of the aforementioned limits on biological accuracy, by examining more general MDA dosage schemes, reconsidering the model dynamics in the presence of immunity and demographic processes, and modelling emerging \textit{P. vivax} concerns, such as antimalarial resistance. We also hope to extend our analytical approximations for the optimal MDA time to account for cases in which the prevalence has not attained a limit cycle, but is rather decaying, as is the case for the parameters in Table \ref{parameters}. 

To our knowledge, no previous analysis has mathematically evaluated \textit{P. vivax} interventions for forest-going populations in particular, even as \textit{P. vivax} is implicated in over 80\% of malaria cases in Cambodia, and forest-going populations account for over half of malaria cases in other regions of Southeast Asia (e.g., Indonesia) \cite{doum2023active,  gallalee2024forest}. Our work therefore provides insight into the complex considerations underpinning forest malaria prevention in Southeast Asia, contributing to a larger body of work combating \textit{P. vivax} endemicity.

\section*{Appendix I: Proof of Theorem \ref{thmEq}}
 \begin{proof}[Proof of Theorem \ref{thmEq}] To prove the theorem, we first introduce the continuous random variables $\mathcal{P}(t), \mathcal{B}, \mathcal{W},$ with support given by the sets $\{p_j(t)\}_{j \in \mathcal{B}}$, $\{B_j(t)\}_{j \in \mathcal{M}}$, and $\{w_{ij}(t)\}_{i,j \in \mathcal{M}}$, respectively, and with constant probability density function given by $\frac{1}{M}$. We note that \[\mathbb{E}[\mathcal{B}(t)] = A_1(t)\] 

and \[\mathcal{B}(t) = u + v \cos(\omega t + \mathcal{W}).\]
Letting $\textbf{1}_{I_M, j}(t)$ be the indicator function equal to $1$ if $j \in I_M$ at time $t$ and to $0$ otherwise, we have that 
\begin{equation}\label{eq:exp} \begin{aligned} I_M(t) = \int_{j \in M} \mathbb{E}\bigr[\textbf{1}_{I, j}(t)] \, dj= \int_{j \in M} p_j(t) \, dj. \end{aligned} \end{equation}

Moreover,  \begin{equation} \begin{aligned} & I_M(t) = M \mathbb{E}\bigr[ \mathcal{P}(t) \bigr] = \\ & M - M\mathbb{E}\Bigr[\exp\pa{\int_{0}^{t} \pa{\mathcal{B}(\tau) \lambda_1(\tau) + (1 - \mathcal{B}(\tau)) \lambda_2(\tau)}f(t - \tau)\ d\tau} \Bigr],\end{aligned}\end{equation}
where, with the assumption that $\mathcal{W}$, which is bounded from above by $\textnormal{sup}_{i,j\in\mathcal{M}} \abs{w_{i,j}}$, is small (in particular, that $\mathcal{W} << 1 - u - v \cos(\omega t)$ for all $t$), we obtain that 
\begin{equation} \begin{aligned} \int_{0}^{t} \pa{\mathcal{B}(\tau) \lambda_1(\tau) + (1 - \mathcal{B}(\tau)) \lambda_2(\tau)}f(t - \tau)\ d\tau \approx \frac{ab}{M}\,\int_{0}^{t}I_{m2}(\tau) f(t - \tau)\ d\tau \, + \\ u\int_{0}^{t}  \lambda_1(\tau)f(t - \tau)\ d\tau + v \int_{0}^{t} \lambda_1(\tau) \cos(
\omega \tau - \mathcal{W}) f(t - \tau)\ d\tau, \end{aligned} \end{equation} 
where we have used the fact that \be \begin{aligned} \int_0^t \pa{1 - \mathcal{B}(\tau)} \lambda_2(\tau) f(t - \tau) \, d\tau = \int_0^t \pa{1 - \mathcal{B}(\tau)} \pa{\frac{ab\,I_{m2}(\tau)}{(1 -  u - v\cos (\omega \tau))M}} f(t - \tau) \, d\tau \\ \approx \frac{ab}{M}\,\int_{0}^{t}I_{m2}(\tau) f(t - \tau)\ d\tau. \end{aligned} \ee

For small $Q_1, \abs{v}, \textnormal{sup}_{i,j\in\mathcal{M}} \abs{w_{i,j}}$ and large $\omega$, (e.g., $\omega >> 1$), \[\mathcal{C}(t):=v  \int_{0}^{t} \frac{ab\, I_{m1}(\tau) \cos(\omega \tau - \mathcal{W})}{N + (u + v \cos \omega t) \, M} f(t - \tau)\ d\tau\] 
satisfies \[\abs{\mathcal{C}(t)} << 1,\] where we use the fact that $|f(t)|$ is bounded by $1$ for all $t$.

We note that 
\be \begin{aligned} \abs{\mathbb{E}\bigr[1 - \mathcal{P}(t)  \bigr] - \exp \pa{\mathbb{E}\Bigr[\int_{0}^{t} \pa{\mathcal{B}(\tau) \lambda_1(\tau) + (1 - \mathcal{B}(\tau)) \lambda_2(\tau)} f(t - \tau)\ d\tau  \Bigr]}} <  \\ \max \mathcal{C}^2(t) \cdot  \exp\pa{\frac{ab}{M}\,\int_{0}^{t}I_{m2}(\tau) f(t - \tau)\ d\tau \, + u\int_{0}^{t}  \lambda_1(\tau)f(t - \tau)\ d\tau }.\end{aligned} \ee

Since \[\mathbb{E}\bigr[1 - \mathcal{P}(t)  \bigr] = e^{\mathcal{C}(t)} \exp\pa{\frac{ab}{M}\,\int_{0}^{t}I_{m2}(\tau) f(t - \tau)\ d\tau \, + u\int_{0}^{t}  \lambda_1(\tau)f(t - \tau)\ d\tau },\]
we see that
\be \begin{aligned} \abs{\mathbb{E}\bigr[1 - \mathcal{P}(t)  \bigr] - \exp \pa{\mathbb{E}\Bigr[\int_{0}^{t} \pa{\mathcal{B}(\tau) \lambda_1(\tau) + (1 - \mathcal{B}(\tau)) \lambda_2(\tau)}f(t - \tau)\ d\tau  \Bigr]}} <  \\ \frac{\mathcal{C}^2(t)}{e^{\mathcal{C}(t)}}  \ \mathbb{E}\bigr[1 - \mathcal{P}(t)  \bigr].\end{aligned} \ee 

This implies that we can choose some $\abs{\epsilon} \le \frac{\mathcal{C}^2(t)}{e^{\mathcal{C}(t)}}$ such that \[\mathbb{E}\bigr[1 - \mathcal{P}(t)  \bigr] = \frac{1}{1+\epsilon}\Biggr[\exp \pa{\mathbb{E}\Bigr[\int_{0}^{t} \pa{\mathcal{B}(\tau) \lambda_1(\tau) + (1 - \mathcal{B}(\tau)) \lambda_2(\tau)}f(t - \tau)\ d\tau  \Bigr]}\Biggr]\]

Since \[I_M(t) = M - M\mathbb{E}\bigr[1 - \mathcal{P}(t) \bigr],\] we find that \be \begin{aligned} & I_M(t) = M - M \pa{\frac{1}{1 + \epsilon}} \times \\ & \exp\pa{\int_{0}^{t} \pa{\mathbb{E}\bigr[ \mathcal{B}(\tau) \bigr] \lambda_1(\tau) + (1 - \mathbb{E}\bigr[ \mathcal{B}(\tau) \bigr]) \lambda_2(\tau)}f(t - \tau)\ d\tau}.\end{aligned} \ee
In particular, \[\lim_{\mathcal{C} \to 0} \frac{1}{1+\epsilon} = 1,\] which yields Theorem \ref{thmEq}.
\end{proof}

\section*{Appendix II: Proof of Theorem \ref{thmMDA}}
\begin{proof}[Proof of Theorem \ref{thmMDA}]
To derive a solution for the optimization problem, we will use the following two results: \\ 
\begin{lem}\label{lemma1} For sufficiently-large $\omega$ and $p_{\textnormal{blood}}, p_{\textnormal{rad}}$ arbitrarily close to $1$, a value of $s$ that minimizes $AP(s)$ must maximize \[\int_0^{\frac{2 \pi}{\omega}} q_I(s+k, s) \ dk,\] where $\tilde{I}_M(s + k, s) = 1 - e^{q_I(s+k, s)}$. \end{lem}

\begin{lem}\label{lemma2} Fixing $t \in \mathbb{R}_+$, we have that \[\abs{\frac{\partial }{\partial s} r(s, t)} << |r(s, t)| \ \text{and} \ \abs{\frac{d }{d s} f(s)} << |f(s)|,\] for all  $s \in \mathbb{R}_+$.  \end{lem}

Proofs for Lemmas \ref{lemma1} and \ref{lemma2} are provided in Appendices III and IV, respectively. Given these results, we consider values of $s \in [T, T + 2\pi/\omega]$ for which \[\int_0^{\frac{2 \pi}{\omega}} \frac{\partial q_I(s+k, s)}{\partial s} \ dk = 0.\]
Taking the partial derivative of Equation \eqref{pIts} with respect to $s$ and noting that $\tilde{I}_M(t, s) \neq 1$, we obtain that the value of $s$ that minimizes $AP(s)$ must satisfy 
\be\label{equalszero}\begin{aligned} \int_0^{\frac{2 \pi}{\omega}} \pa{h(s)\,r(0, k)+\int_0^s h(\tau) \, r'(s-\tau, s+k-\tau) \, d\tau}\, dk \\ + \int_0^{\frac{2 \pi}{\omega}} \pa{h(s+k)\,f(0) - h(s)\,f(k) + \int_s^{s+k} h(\tau) \,f'(s+k-\tau) \, d\tau}\, dk = 0, \end{aligned}\ee
where
\be \begin{aligned} & h(t) = \lambda_{1}(t)A_1(t) + \lambda_{2}(t)A_2(t) = \\ & \frac{I_{m1}(t)(u + v \cos \,(\omega t))}{N + (u + v \cos \,(\omega t))M}  \, + \frac{I_{m2}(t)((1 - u) - v\cos \,(\omega t))}{((1 - u) - v\cos \,(\omega t)M}.\end{aligned} \ee
When $v$ is arbitrarily small, \[h(t) \approx \frac{uabI_{m1}(\tau) }{N + Mu}  \, + \frac{abI_{m2}(\tau)}{M}.\]
Moreover, using Equation \eqref{coeff}, we find that \be\begin{aligned} h(t) \approx \frac{uab \pa{w^{(1)}_{-1} e^{-\omega \textbf{i} t} + w^{(1)}_0 + w^{(1)}_{1} e^{\omega \textbf{i} t}}}{N+Mu} + \frac{ab\pa{w^{(2)}_{-1} e^{-\omega \textbf{i} t} + w^{(2)}_0 + w^{(2)}_{1} e^{\omega \textbf{i} t}}}{M}, \end{aligned} \ee
so, fixing $k$, \be\begin{aligned} \int_0^s h(\tau) \, r'(s-\tau, s+k-\tau) \, d\tau = \\ \pa{\frac{Muabw^{(1)}_0 + (N+Mu)abw^{(2)}_0}{M(N+Mu)} }\pa{r(s, s + k) - r(0, k)} + \\  \pa{\frac{Muabw^{(1)}_{1} + (N+Mu)abw^{(2)}_{1}}{M(N+Mu)} \cdot \, e^{\omega \textbf{i} s}\int_0^s e^{-\omega \textbf{i} \tau} r'(\tau, \tau+k) \, d\tau } + \\ \pa{\frac{Muabw^{(1)}_{-1} + (N+Mu)abw^{(2)}_{-1}}{M(N+Mu)} \cdot \, e^{-\omega \textbf{i} s} \int_0^s e^{\omega \textbf{i}\tau} r'(\tau, \tau+k) \, d\tau }\end{aligned} \ee 

and \be\begin{aligned} \int_s^{s+k} h(\tau) \, f'(s+k-\tau) \, d\tau = \pa{\frac{Muabw^{(1)}_0 + (N+Mu)abw^{(2)}_0}{M(N+Mu)} }\pa{f(k) - f(0)} + \\  \pa{\frac{Muabw^{(1)}_{1} + (N+Mu)abw^{(2)}_{1}}{M(N+Mu)} \cdot \, e^{\omega \textbf{i} s}\int_{0}^{k} e^{-\omega \textbf{i} \tau} f'(\tau+k) \, d\tau } + \\ \pa{\frac{Muabw^{(1)}_{-1} + (N+Mu)abw^{(2)}_{-1}}{M(N+Mu)} \cdot \, e^{-\omega \textbf{i} s} \int_{0}^{k} e^{\omega \textbf{i}\tau} f'(\tau+k) \, d\tau }.\end{aligned} \ee 
Given that $w^{(1)}_{1}, w^{(1)}_{-1}, w^{(2)}_{1}, w^{(2)}_{-1} << 1$, $f'(t) << 1$, and $\partial_s \, r(s, t) << 1$ by Lemma \ref{lemma2}, one can reduce Equation \eqref{equalszero} to \be\begin{aligned} \int_0^{\frac{2 \pi}{\omega}}  \pa{\frac{Muabw^{(1)}_0 + (N+Mu)abw^{(2)}_0}{M(N+Mu)} }\pa{r(s, s+k) + f(k) - f(0) - r(0, k)} + \\ \pa{h(s)\,r(0, k)+ h(s+k)\,f(0) - h(s)\,f(k)} dk = 0.\end{aligned}\ee
Since \be\begin{aligned} \int_0^{\frac{2 \pi}{\omega}} h(s+k) \, dk  = \frac{2 \pi}{\omega} \pa{\frac{Muabw^{(1)}_0 + (N+Mu)abw^{(2)}_0}{M(N+Mu)} },\end{aligned}\ee  we obtain \be\begin{aligned} 0 = h(s) \int_0^{\frac{2 \pi}{\omega}} (r(0, k) - f(k)) \, dk + \\ \pa{\frac{Muw^{(1)}_0 + (N+Mu)abw^{(2)}_0}{M(N+Mu)} } \int_0^{\frac{2 \pi}{\omega}}  \pa{r(s, s+k) + f(k) - r(0, k)} \, dk = \\ 
\pa{\frac{Muab(w^{(1)}_{-1}e^{-\omega \textbf{i} s} + w^{(1)}_{1}e^{\omega \textbf{i} s}) + (N+Mu)ab(w^{(2)}_{-1}e^{-\omega \textbf{i} s} + w^{(2)}_{1}e^{\omega \textbf{i} s})}{M(N+Mu)}} \times \\ \int_0^{\frac{2 \pi}{\omega}} (r(0, k) - f(k)) \, dk \, + \\ \pa{\frac{Muabw^{(1)}_0 + (N+Mu)abw^{(2)}_0}{M(N+Mu)} } \int_0^{\frac{2 \pi}{\omega}}  r(s, s+k) \, dk. \end{aligned}\ee

We claim that this equation always has a solution for $s \in [T, T + 2\pi/\omega]$. In particular, using the fact that $s$ is large, so that the right-hand side of Equation \eqref{req} 
tends to zero, we obtain that the approximate solution for $s-T$ must satisfy
\be \begin{aligned} & \tan\,(\omega (s-T)) = -\frac{Mu(w^{(1)}_{-1} + w^{(1)}_{1}) + (N+Mu)(w^{(2)}_{-1} + w^{(2)}_{1})}
{Mu (w^{(1)}_{1} - w^{(1)}_{-1})\textbf{i}  + (N+Mu)(w^{(2)}_{1} - w^{(2)}_{-1})\textbf{i} }
\end{aligned}
\ee and \be s - T \in [0, 2\pi/\omega). \ee

The solution set to Equation \eqref{soleq} contains both the maximum and the minimum of $\int_0^{2\pi/\omega} q_I(s+k,s) \, dk$. To isolate the maximum, we impose the additional constraint 
\be\begin{aligned}
-\pa{Mu(w^{(1)}_{-1} + w^{(1)}_{1}) + (N+Mu)(w^{(2)}_{-1} + w^{(2)}_{1})} \sin (\omega (s-T)) + \\ \pa{Mu (w^{(1)}_{1} - w^{(1)}_{-1})\textbf{i}  - (N+Mu)(w^{(2)}_{1} - w^{(2)}_{-1})\textbf{i} } \cos \,(\omega (s-T)) < 0,
\hspace{10pt} 
\end{aligned}
\ee yielding the result in Theorem \ref{thmMDA}.

In particular, the optimal time of MDA application is explicitly dependent only on the proportion of the moving and stationary populations (represented by $M$ and $N$), the proportion of infectious mosquitoes (the oscillation in which is represented by $w^{(1)}_{-1}, w^{(1)}_{1}, w^{(2)}_{-1},$ and $w^{(2)}_{1}$), and the mean movement pattern (represented by $u$, $v$, and $\omega$). The dependence on $p_{\textnormal{blood}}, p_{\textnormal{rad}}$ is small, and lost when the minor terms (which are of order $O(w^{(1)}_1 + w^{(1)}_2)$) in Equations (30) and (31) are set to zero.  
\end{proof}
\section*{Appendix III: Proof of Lemma \ref{lemma1}}
\begin{proof}[Proof of Lemma \ref{lemma1}] To show that the value of $s$ maximizing $\int_0^{\frac{2 \pi}{\omega}} e^{q_I(s+k, s)} \ dk$ is arbitrarily close to the value maximizing 
$\int_0^{\frac{2 \pi}{\omega}} \pa{1+q_I(s+k, s)} \ dk$, we note that \be\begin{aligned}
 \lim_{p_{\textnormal{blood}},\,  p_{\textnormal{rad}} \, \to \, 1} q_I(s+k, s) = \int_s^{s+k} h(\tau) f(s+k - \tau) \, d\tau = \\ \pa{\frac{Muabw^{(1)}_0 + (N+Mu)abw^{(2)}_0}{M(N+Mu)} \int_0^{k} f(\tau) \, d\tau} +\\ \int_s^{s+k} e^{-\omega \textbf{i} \tau} \pa{\frac{Muabw^{(1)}_{-1} + (N+Mu)abw^{(2)}_{-1}}{M(N+Mu)}} f(s+k-\tau) \, d\tau + \\ \int_s^{s+k} e^{\omega \textbf{i} \tau}\pa{\frac{Muabw^{(1)}_1 + (N+Mu)abw^{(2)}_1}{M(N+Mu)}} f(s+k-\tau) \, d\tau,
\end{aligned} \ee where the $s$-dependent term is very small and \be \begin{aligned}\abs{\int_0^k f(\tau) \, d\tau} \le \abs{\int_0^{\frac{2\pi}{\omega}} f(\tau) \, d\tau} \le  \int_0^{\frac{2\pi}{\omega}} \pa{e^{-\gamma \tau} + \nu p_A(\tau)} \, d\tau \le \\ \frac{2\alpha+\mu-\gamma}{\gamma(\alpha + \mu - \gamma)}\pa{1-e^{-2\pi\gamma/\omega}} - \frac{\alpha}{(\alpha+\mu)(\alpha + \mu - \gamma)}\pa{1-e^{-2\pi(\alpha + \mu)/\omega}} << 1\end{aligned}
\ee for sufficiently large $\omega$ and standard values of $\alpha, \mu, \gamma$ (see Table \ref{parameters}). This implies that \[\lim_{p_{\textnormal{blood}},\,  p_{\textnormal{rad}} \, \to \, 1} q_I(s+k, s) << 1,\] from where the lemma follows. 
 \end{proof}

 \section*{Appendix IV: Proof of Lemma \ref{lemma2}}
\begin{proof}[Proof of Lemma \ref{lemma2}] We first show that $|f'(t)| << |f(t)|$. In particular, we note that \[-p_A'(t) < \max\{\alpha + \mu, \gamma\}p_A(t),\] where both $\alpha + \mu$ and $\gamma$ are assumed to be very small (on the order of $10^{-2}$, as shown in Table \ref{parameters}). 
As such, \be \begin{aligned} |f'(t)| = \abs{\frac{\gamma e^{-\gamma t}}{1 + \nu p_A(t)} - \frac{\nu p'_A(t) \pa{1- e^{-\gamma t}}}{\pa{1 + \nu p_A(t)}^2}} \le \\ \abs{\frac{\gamma\pa{1 + \nu p_A(t)} e^{-\gamma t}}{\pa{1 + \nu p_A(t)}^2} + \frac{\nu\, \max\{\alpha + \mu, \gamma\}\, p_A(t)}{\pa{1 + \nu p_A(t)}^2}} << |f(t)|, \end{aligned} \ee 
where in the last step we have used the fact that \[\gamma(1 + \nu p_A(t)) < \gamma(1 + \nu) << 1\] when $\nu$ is assumed to be much smaller than $1/\gamma$, as in Table \ref{parameters}.
The proof for $r(t,s)$ follows similarly.
\end{proof}
 
\section*{Funding}
 J.A. Flegg’s research is supported by the Australian Research Council (FT210100034, CE230100001) and the National Health and Medical Research Council (APP2019093). 

\section*{Data availability}
Data extracted electronically from Wangdi \textit{et al.} \cite{wangdi2018analysis} is included in the Github repository at \href{https://github.com/jflegg/malaria-three-scale}{https://github.com/jflegg/malaria-three-scale}.

\appendix

\FloatBarrier
\bibliographystyle{plain}
\bibliography{bibliography}
 
\end{document}